\documentclass[lettersize,journal]{IEEEtran}
\usepackage{cite}
\usepackage{amsmath,amssymb,amsfonts}

\usepackage{algorithm} 
\usepackage{textcomp}
\usepackage{bm}
\usepackage{amsmath,amsfonts}
\usepackage{algorithmic}
\usepackage{array}
\usepackage[caption=false,font=normalsize,labelfont=sf,textfont=sf]{subfig}
\usepackage{textcomp}
\usepackage{stfloats}
\usepackage{url}
\usepackage{verbatim}
\usepackage{graphicx}
\usepackage{color}

\newtheorem{lemma}{Lemma}
\newenvironment{proof}{{\noindent\it \quad Proof:}\quad}{\hfill $\square$\par}

\def\BibTeX{{\rm B\kern-.05em{\sc i\kern-.025em b}\kern-.08em
    T\kern-.1667em\lower.7ex\hbox{E}\kern-.125emX}}
\usepackage{balance}

\begin{document}
\title{Two-Timescale Synchronization and Migration for Digital Twin Networks: A Multi-Agent Deep Reinforcement Learning Approach}

\author{ {Wenshuai Liu,~\IEEEmembership{Student Member,~IEEE}, Yaru Fu,~\IEEEmembership{Member,~IEEE},~Yongna Guo,~\IEEEmembership{Member,~IEEE}\\
Fu Lee Wang,~\IEEEmembership{Senior Member,~IEEE},~Wen~Sun,~\IEEEmembership{Senior Member,~IEEE},~and~Yan Zhang,~\IEEEmembership{Fellow,~IEEE}}
	\thanks{ 
   	This work was supported in part by the Hong Kong Research Matching Grant (RMG) in the Central Pot under Project No. CP/2022/2.1, in part by the Research and Development Fund (R\&D Fund) under reference No. RD/2023/1.8,  in part by the Team-based Research Fund under Project No. TBRF/2024/1.10, and in part by the EU Horizon 2020 Research and Innovation Programme under the Marie Sklodowska-Curie grant agreement No. 101008297. This article reflects only the authors' view. The European Union Commission is not responsible for any use that may be made of the information it contains. \emph{(Corresponding author: Yaru Fu)} 
	
		W. Liu, Y. Fu, and F.-L. Wang are with the School of Science and Technology,  Hong Kong Metropolitan University, Hong Kong, 999077, China (e-mail: liuws1996@gmail.com; yfu@hkmu.edu.hk; pwang@hkmu.edu.hk).

Y. Guo was with the School of Science and Technology,  Hong Kong Metropolitan University, Hong Kong, 999077, China (e-mail: yongnaguo2-c@my.cityu.edu.hk).

W. Sun is with the School of Cybersecurity, Northwestern Polytechnical University, 127 West Youyi Road, Xi’an Shaanxi, 710072, China (e-mail: sunwen@nwpu.edu.cn)

  Y. Zhang is with the Department of Informatics, University of Oslo (e-mail: yanzhang@ieee.org). 
		
	}
}

\maketitle

\begin{abstract}
Digital twins (DTs) have emerged as a promising enabler for representing the real-time states of physical worlds and realizing self-sustaining systems. In practice, DTs of physical devices, such as mobile users (MUs), are commonly deployed in multi-access edge computing (MEC) networks for the sake of reducing latency. To ensure the accuracy and fidelity of DTs, it is essential for MUs to regularly synchronize their status with their DTs. However, MU mobility introduces significant challenges to DT synchronization. Firstly, MU mobility triggers DT migration which could cause synchronization failures. 
Secondly, MUs require frequent synchronization with their DTs to ensure DT fidelity. Nonetheless, DT migration among MEC servers, caused by MU mobility, may occur infrequently. Accordingly, we propose a two-timescale DT synchronization and migration framework with reliability consideration by establishing a non-convex stochastic problem to minimize the long-term average energy consumption of MUs. 
We use Lyapunov theory to convert the reliability constraints and reformulate the new problem as a partially observable Markov decision-making process (POMDP). 
Furthermore, we develop a heterogeneous agent proximal policy optimization with Beta distribution (Beta-HAPPO) method to solve it. Numerical results show that our proposed Beta-HAPPO method achieves significant improvements in energy savings when compared with other benchmarks.
\end{abstract}


\section{Introduction}
\label{sec:introduction}
The rapid development of the Internet of Things (IoT) and wireless communication technologies has spurred the emergence of numerous new applications and services, such as smart cities, multidimensional sensing, and ubiquitous intelligence \cite{Yang2019IN_6G, Letaief2022JSAC, yfu2023}. 
However, these new applications and services pose new demands and challenges on the network, which may include ubiquitous connectivity, ultra-reliable and low-latency communications, and massive device connectivity. 
To address these challenges, it is essential to develop a self-sustaining paradigm and a proactive, online learning-based wireless networks \cite{Wu2021IOTJ_Digital,Wang2022TCOMR}. 

Wireless digital twin networks (WDTNs) are designed to effectively address these challenges.
WDTN is defined as digital frameworks based on digital twin (DT) technology, creating high-fidelity virtual models of physical wireless network. This allows for real-time monitoring, simulation, and analysis to enhance connectivity and performance across various applications and services. This innovation is anticipated to be a driving force in the digitalization of the physical world, facilitating seamless communication and integration between the physical and virtual domains\cite{Tang2022OJCOMS,Wang2023IOTJ1,Gu2024TVT}. 
As a high-fidelity digital replica of the physical network throughout its lifecycle, WDTN plays a crucial role in ensuring network performance and reliability.
DT-driven networks continuously monitor, simulate, and analyze the physical world, allowing for ongoing assessment and adjustment of the efficiency and effectiveness of quality of service policies. By leveraging these capabilities, WDTN can optimize and predict the behavior of physical systems, ensuring optimal performance and promoting proactive decision-making within the network ecosystem. This predictive and optimization capability provides network operators with unprecedented management and maintenance capabilities, significantly enhancing network operational efficiency and user satisfaction \cite{Lu2021IOTJ, Alcaraz2022COMST_Digital}. 
Moreover, the application of advanced machine learning (ML) techniques in analyzing digital networks empowers DTs to significantly enhance optimization and decision-making processes in physical networks \cite{ Wang2022IoT_Mobility, Li2024TMC}.

WDTN is a system composed of three layers, namely, the physical layer, digital twin  layer, and application layer. The physical layer typically consists of terminal devices and communication infrastructure. The DT layer is responsible for creating and maintaining virtual replicas of the physical layer's components to enable real-time monitoring and analysis. The application layer handles the implementation of advanced functions and services, such as data analysis and decision support. This three-layer architecture allows the WDTN to efficiently integrate and coordinate interactions between the physical and digital worlds \cite{Duong2023MWC}.
To establish WDTN, it is crucial to carefully plan the deployment of DTs within the network infrastructure.
While deploying DTs in remote clouds is possible, it can introduce significant communication delays and compromise fidelity. The emergence of multi-access edge computing (MEC) networks, which leverage multiple edge servers for communication and computation with end devices, offers a realistic solution.
Deploying DTs in MEC networks helps alleviate communication delays \cite{Duong2023MWC,Tang2022OJCOMS, Zhang2023JSAC}.

Moreover, DT in MEC networks can also facilitate decision-making processes related to computation offloading and resource allocation through real-time interaction between the digital world and the physical space, giving rise to DT-assisted MEC\cite{Khan2022MCOM, Kurma2023TCOM}.
Several notable research studies have explored the role of DTs in optimizing different aspects of MEC networks\cite{ Dai2021TII,Li2022TVT, Van2022TCOM, Van2023JSAC, Xu2023JSAC,Liu2022IOTJ}. 
Specifically, the work \cite{Dai2021TII} applied ML techniques in DT to enhance the design of the edge computing networks. By mirroring the dynamic network in DT, ML is applied to make decision on real-time resource allocation, which effectively improve energy efficiency and data processing efficiency.
In \cite{Li2022TVT}, an adaptive aerial edge computing network utilizing DT was proposed, wherein the real-time environmental predictions were leveraged through DT to introduce intelligent offloading strategies for network resource allocation, leading to a significantly reduced network energy consumption.
In \cite{Van2022TCOM} and \cite{Van2023JSAC}, the authors explored ultra-reliable and low-latency communications (URLLC) scenarios within DT-assisted MEC networks. By modeling the edge network with DT, they effectively achieved the goal of minimizing delays under the stringent transmission requirements of URLLC.
In \cite{Xu2023JSAC}, a DT-driven edge-collaborative scheduling algorithm was proposed for delay-sensitive tasks in MEC networks. By effectively coordinating computation and communication resources, the devised method can significantly reduce task completion times.
Additionally, in \cite{Liu2022IOTJ}, the authors integrated DT with blockchain to improve MEC network security. They proposed an intelligent task offloading algorithm to reduce delay and energy consumption while ensuring data security.

DTs are complete replicas of physical devices in the digital space, including hardware configurations, historical running data, and real-time status information \cite{Lu2021IOTJ}. To ensure that DTs accurately reflect the current state of their physical counterparts, a synchronization process is essential. For effective synchronization, DTs require real-time data exchange. The types of data involved typically include observed data such as signal strength, network load, and device status, user data like service usage patterns and preferences, and environmental data that encompasses factors such as geographic location and local interference sources. The quantity of data required for synchronization can be substantial, as it needs to capture the comprehensive state of the network and each connected device at any given moment \cite{Guo2024MWC,Yu2024IOTM,Han2023JIOT}.
Moreover, the process of collecting a large amount of data to maintain synchronization between the DT and its physical object is considered both compute-intensive and data-intensive. Maintaining synchronization of the DT consumes a significant amount of computation resources. Therefore, effective resource allocation is crucial to achieve efficient synchronization between the DT and its physical object \cite{Li2024TMC}.

Furthermore, it is crucial to carefully consider the deployment of DTs among the edge servers and ensure real-time synchronization between the physical devices and their corresponding DTs. This attention to detail is essential to guarantee the reliability and accuracy of DT networks. 
The work \cite{Lu2021IOTJ,Vaezi2023IOTJ,Li2024TMC,Zheng2023TWC,Chukhno2022IOTJ} examined the DT deployment problems to satisfy different requirements.
Specifically, in \cite{Vaezi2023IOTJ}, the authors devised a model for DT synchronization that addresses the DT deployment problem, to minimize the age of information associated with DT data.
The authors in \cite{Li2024TMC} proposed a DT synchronization model based on continual learning. This model optimizes computation resources to maximize the total utility gain of the DT network, thereby enhancing the accuracy of digital models. 
The aforementioned studies \cite{Vaezi2023IOTJ,Li2024TMC} have not considered user mobility and networks dynamics, which may significantly impact the synchronization of DTs. 
To address this challenge, the authors in \cite{Zheng2023TWC} investigated DT synchronization problem in vehicular networks, which analyzed the association between DT and wireless access points to ensure low-latency synchronization.
Furthermore, in \cite{Chukhno2022IOTJ}, the authors took into consideration social characteristics, computation resources, and mobile user (MU) mobility when analyzing DT deployment. They proposed an approximation algorithm to minimize the synchronization delay. 
Moreover, it is noteworthy that to meet the stringent requirements for reliability and low latency, DT migration can be triggered due to MU mobility. 
To address this issue, the authors in \cite{Lu2021IOTJ} investigated DT deployment and migration using ML techniques. They aimed to maximize MU utility by achieving a balance between latency and energy costs.

While previous studies have explored DT synchronization and migration resulting from MU mobility, the issue of DT synchronization failure caused by DT migration has not been taken into account.
We acknowledge the importance of MU mobility and its potential impact on DT migration and subsequent synchronization failures, which calls for further research. Additionally, to ensure the accuracy of DTs, physical devices must frequently synchronize data with their corresponding DTs, necessitating synchronization on a small timescale. 
However, when operating on the same timescale, MUs located at the network edge may frequently enter and exit the network, which can increase the volatility of the network and potentially lead to frequent migrations of the DTs. Additionally, the migration of the DT requires transferring all the data contained within the DT and reconfiguration after the migration. Such frequent migrations not only consume substantial resources but also cause frequent service interruptions, thereby hindering the effective synchronization of DTs. Therefore, to minimize network overhead and reduce the likelihood of synchronization failures, it is of high necessity to plan and execute DT migrations on a larger timescale. By planning and executing DTs migration on a larger timescale, resources can be managed more effectively, interruptions can be reduced, thereby enhancing the reliability of the network.
Hence, this paper aims to tackle the challenges associated with synchronization and migration within DT networks. To achieve this goal, we present a novel two-timescale framework, focusing on minimizing 
long-time average energy consumption while adhering to a specified synchronization failure constraint. The main contributions of this work are summarized below:

\begin{itemize}
    \item  
   We propose a two-timescale framework that incorporates both small timescale DT synchronization and large timescale DT migration.
   In addition, we consider the possibility of DT synchronization failure during the migration process. 
   We then jointly optimize DT deployment, communication, and computation resources, while minimizing long-term average energy consumption under the constraints of DT synchronization reliability.
    \item  
    The formulated problem is challenging to solve due to its non-convex and stochastic nature. 
    Therefore, we exploit the multi-agent deep reinforcement learning (MADRL) method to tackle this problem.
    To utilize MADRL, we first transform the DT reliability constraints into tractable constraints using the Lyapunov optimization theory. Then, we reformulate the newly generated optimization problem as a partially observable Markov decision process (POMDP). The MADRL framework with heterogeneous agents is applied to decompose the high-dimensional state and action spaces. Whilst, the heterogeneous agent proximal policy optimization (HAPPO) method is used to determine the optimization strategy of the network.
    \item 
    Our two-timescale model faces training challenges due to varying sampling efficiencies among agents. The control center operates on a large timescale with lower sample efficiency, which slows training. To mitigate this, we enhance training efficiency by collecting additional rewards for the control center on a smaller timescale without impacting other agents. Additionally, the use of Gaussian distribution in the action outputs introduces bias due to the bounded action space. Therefore, we utilize a Beta distribution in the actor network for policy output.  
    Extensive numerical results show that the HAPPO with Beta distribution (Beta-HAPPO) algorithm possesses efficient convergence and effectively minimizes energy consumption.
\end{itemize}

The following sections of this paper are organized as follows.  Section \ref{s:sys} details the system model and introduces the long-term average energy minimization problem under consideration. In Section \ref{s:mdp}, the original problem is transformed by Lyapunov optimization theory and described as a POMDP.
Section \ref{s:proposed} elaborates on the proposed Beta-HAPPO distributed MADRL algorithm.
Numerical evaluations and discussions are provided in Section \ref{s:simulation}. Finally, we conclude this paper and predict future research directions in Section \ref{s:conclusion}.

\begin{figure}[t]
	\centering
	\includegraphics[width=0.5\textwidth]{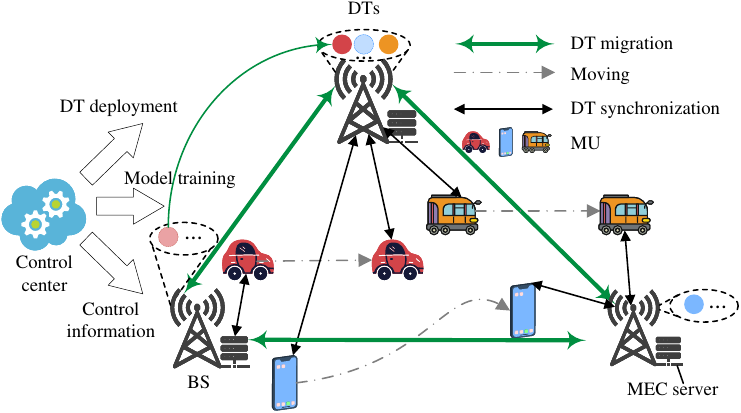}
 \vspace{-2em}
	\caption{ An illustration of the DT-empowered MEC network.}
	\label{fig:sys}
         \vspace{-1.5em}
\end{figure}

\section{System Model and Problem Formulation}\label{s:sys}
We consider a DT-empowered MEC network as shown in Fig. \ref{fig:sys}, which consists of a control center, $M$ single-antenna base stations (BSs) and $K$ single-antenna MUs. The control center is used to maintain the model training and the deployment decision of DTs. Moreover, each BS is equipped with an MEC server to assist MUs in constructing and synchronizing their DTs, and  DTs of MUs have been initialized in MEC servers. We define the MU set as $\mathcal{K}=\{1,2,\ldots,K\}$ and the set of BSs as $\mathcal{M}=\{1,2,\ldots,M\}$.
The network functions by dividing time into slots, where each time slot has length $\delta_t$ that is equal to the coherent time of the wireless channel.
The set of time slots is represented by $\mathcal{N}=\{0,1,2,\ldots,n,\ldots\}$.
The position of MU $k$ during time slot $n$ within the Cartesian plane is expressed as ${\bf u}_k[n]=[x_{u,k}[n],y_{u,k}[n]]^{\text{T}}$, and BS $m$ is ﬁxed at location ${\bf q}_m=[x_{b,m},y_{b,m}]^{\text{T}}$. In this paper BS $m$ and MEC server $m$ are used substitutively for each other.

Consider that the movement of MUs typically exhibits temporal dependency, namely, their current position and speed are dependent on their previous states.
 The Gaussian-Markov model, a stochastic process with memory, effectively captures this time-related dependency and is straightforward to implement. 
 Therefore, we employ the Gaussian-Markov model to characterize the mobility of MUs\cite{Liang2003TNET,Batabyal2015COMST,Liu2020TVT}. Within time slot $n$, the velocity $v_k[n]$ and the movement direction $\theta_k[n]$ of MU $k$ are determined by the subsequent equations:
\begin{align}
	 & v_k[n] = \mu_1 v_k \left[ n-1 \right] + \left( 1- \mu_1 \right) \bar s + \sqrt{1 - \mu_1^2} \Phi_k,                \\
	 & \theta_k [n] = \mu_2 \theta_k \left[ n -1 \right]+ \left( 1- \mu_2 \right) \bar \theta + \sqrt{1- \mu_2^2} \Psi_k,
\end{align}
where $\bar s$ is the average speed of MU, $\bar \theta$ is the average moving direction of MU, $0\leq \mu_1,\mu_2\leq 1$ is the influence of MUs' previous state on the current state, $\Phi_k$ and $\Psi_k$ follow the Gaussian distribution satisfying the mean and variance of $(\bar \xi_{v_k}, \zeta_{v_k}^2)$ and $(\bar \xi_{\theta_k}, \zeta_{\theta_k}^2)$, respectively. Therefore, MU $k$' location is updated as follows:

\begin{align}
	 & x_{u,k} [n] = x_{u,k} \left[ n-1 \right] + v_k \left[ n -1 \right] \cos \left( \theta_k \left[ n-1\right]\right) \delta_t, \\
	 & y_{u,k} [n] = y_{u,k} \left[ n-1 \right] + v_k \left[ n -1 \right] \sin \left( \theta_k \left[ n-1\right]\right) \delta_t.
\end{align}

\subsection{Two-Timescale Operation Model}
To ensure the real-time performance of the DT, the MEC servers in BSs require frequent synchronization of information from the MUs. Furthermore, the MUs may switch their BS association after moving for a period of time, making it necessary to migrate the DT. 
Meanwhile, to avoid dense network overhead, we consider DT migration over a large timescale. Thus, we adopt a two-timescale operation model to represent the different frequencies of DT synchronization and DT migration.

As shown in Fig. \ref{fig:timescale}, $T$ sequential time slots are aggregated to form a single time frame. We define the set of time frames as $\mathcal{Q}=\{0,1,2,\ldots,q,\ldots,Q,\ldots\}$, and for the $q$-th frame we denote the collection of time slots as $\mathcal{T}_q=\{qT,qT+1,\ldots,(q+1)T-1\}$. Accordingly, the time frames for DT migration are large timescales, and the time slots for DT synchronization are small timescales. The assumption for timescales is illustrated as follows:
\begin{itemize}
	\item Small timescale: The MUs may require DT synchronization during each time slot. Moreover, in order to maintain DT, joint control of the transmission power of MUs, transmission rate of the BS, and computation resource allocation of the MEC server need to be implemented.
	\item Large timescale: To maintain real-time performance of DT synchronization, DT may migrate themselves from one MEC server to another according to the decisions of the control center at the start of each frame.
\end{itemize}

\begin{figure}[t]
	\centering
	\includegraphics[width=0.5\textwidth]{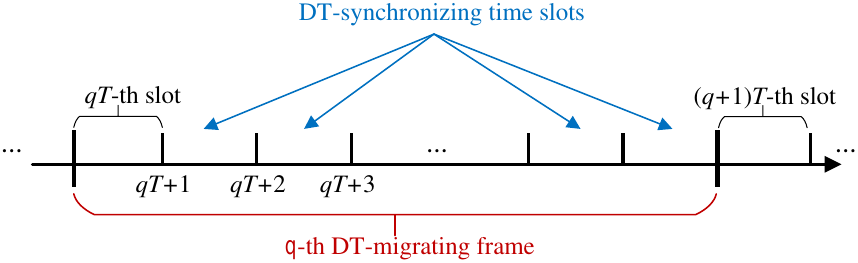}
    \vspace{-2em}
	\caption{An illustration of the two-timescale model of DT migration and DT synchronization.}
	\label{fig:timescale}
     \vspace{-1.5em}
\end{figure}

\subsection{Communications Model}
\subsubsection{MU-to-BS transmission model}
For the uplink transmission, MU $k$ is associated with one BS. The association between MU $k$ and BS $m$ during time slot $n$ is indicated by $\alpha_{k,m}[n]$.  
The transmission link quality directly correlates with the proximity between the MU and the BS, incentivizing MU $k$ to link with the nearest BS for optimal transmission quality.
Therefore, if MU $k$ is closest to BS $m$, we have $\alpha_{k,m}[n] = 1$, otherwise $\alpha_{k,m}[n] = 0$. The transmit power of MU $k$ at time slot $n$ is represented by $p_k[n]$.
In practical implementation, the channel that exists between the MU and the BS is typically modeled by a Rician fading channel, which is characterized as follows \cite{Hua2021TWC}:
\begin{equation}
	{h}_{k,m}[n] = \sqrt{\frac{\rho_0}{d_{k,m}^{\phi}[n]} }\left( \sqrt{ \frac{\kappa}{\kappa +1} } {\bar h}_{k,m}[n] +\sqrt{\frac{1}{\kappa +1}} { \hat h}_{k,m} [n] \right),
\end{equation}
where $\rho_0$ is indicative of the channel power gain at a benchmark distance of 1 meter. 
The path loss exponent is given by $\phi \geq2 $.
We define $d_{k,m}[n]=\Vert {\bf q}_m [n] - {\bf u}_k [n] \Vert_2$ as the Euclidean distance between MU $k$ and BS $m$. 
$\kappa$ reflects the Rician factor that indicates the ratio of power between the direct and indirect signal components. 
${\bar h}_{k,m}[n] \in \mathbb{C}$ corresponds to the line-of-sight (LoS) channel component, where ${\bar h}_{k,m}[n]=1$. ${\hat h}_{k,m}[n] \in \mathbb{C}$ pertains to the non-line-of-sight (NLoS) channel component and follows a complex Gaussian distribution with a mean of zero and a variance of one, denoted by ${\hat h}_{k,m}[n] \sim \mathcal{CN}(0,1)$.

Let the transmit signal symbol of MU $k$ be $x_k[n]\sim \mathcal{CN}(0,1)$, the signal received at BS $m$ from MU $k$ is expressed as
\begin{align}
	 & \nonumber {y}_{k,m}[n] = \alpha_{k,m}[n] \sqrt{p_k[n] }{h}_{k,m}[n] x_k[n] + \\&\quad\sum\limits_{i=1,i \ne k}^K{ \sum\limits_{m=1}^M {\alpha_{i,m}[n] {h}_{i,m}[n] \sqrt{p_i[n]}x_i[n]}}+ {n}_0[n],
\end{align}
where $n_0[n] \in \mathbb{C}$ is the additive white Gaussian noise, with $n_0[n] \sim \mathcal{CN} \left( 0,\sigma^2 \right)$. The signal-to-interference-plus-noise ratio (SINR) received at BS $m$ is obtained as
\begin{equation}
	\gamma_{k,m}[n] = \frac{ p_k[n] \lvert {h}_{k,m}[n] \rvert^2 }{\sum\limits_{i=1,i \ne k}^K{ \sum\limits_{m=1}^M {\alpha_{i,m}[n] \lvert { h}_{i,m}[n] \rvert^2 p_i[n] } }+\sigma^2}.
\end{equation}

Accordingly, the transmission rate of MU $k$ is represented as
\begin{equation}
	R_{k}[n] = \sum_{m=1}^{M}\alpha_{k,m}[n]B \log_2 \left( 1+\gamma_{k,m}[n] \right),
\end{equation}
where $B$ is the communication bandwidth.

\subsubsection{BS-to-BS transmission model}
In the scenario where the DT of MU $k$ is deployed on MEC server $j$, and the MU's associated BS $i$ satisfies $i\neq j$, the BS $i$ must facilitate the transfer of the DT synchronization data for MU $k$ to BS $j$.  
We introduce a binary variable $\zeta_{k,i,j}[n]$ to represent the relationship between the BS $i$ associated with MU $k$ and the BS $j$ where the DT of MU $k$ is deployed. 
When $\zeta_{k,i,j}[n]=1$ and $i \neq j$, it means that MU $k$ is associated with BS $i$, and the DT of MU $k$ is deployed at BS $j$. This implies that wired communication resources are allocated from BS $i$ to BS $j$ for MU $k$. Otherwise, $\zeta_{k,i,j}[n]=0$.
The wired link connection is established between servers at BS $i$ and BS $j$, and $w_{i,j}$ represents the total transmission rate from BS $i$ to BS $j$. $w_{i,j,k}[n]$ is defined as the allocated transmission rate to transmit MU $k$'s data from BS $i$ to BS $j$. This allocation complies to the constraint that $0\leq \sum_{k=1}^{K} \zeta_{k,i,j}[n]w_{i,j,k}[n]\leq w_{i,j}$. 
For simplicity, the wired transmission rate to transfer the data of MU $k$ between BSs is given by $w_k[n]=\sum_{i=1}^{M}\sum_{j=1}^{M} \zeta_{k,i,j}[n] w_{i,j,k}[n]$.

\subsection{DT Synchronization Model}\label{subsec:syn}
Denote $\beta_k(q)\in\mathcal{M}$ as the MEC server where the DT of MU $k$ is deployed. Each MU may request to synchronize its information with the DT during each time slot, and the generation of requests of each MU follows the Bernoulli distribution with success probability $\lambda_k$, which represents the request generation rate of MU $k$. 
Define $a_k[n]$ as a binary indicator, such that $a_k[n]=1$ signifies that MU $k$ initiates a request during time slot $n$, while $a_k[n]=0$ indicates no request is made. Consequently, the probability that MU $k$ makes a request is denoted by ${\rm Pr}(a_k[n]=1)=\lambda_k$, and the probability of not making a request is ${\rm Pr}(a_k[n]=0)=1-\lambda_k$. 
Specifically, the DT synchronization information of MU $k$'s request at each time slot $n$ can be described as $\Omega_k[n]=(D_k[n], C_k[n],\tau_k[n])$, where $\tau_k[n]$ denotes the latency constraint that needs to be met at the start of the time slot.
To maintain prompt synchronization of DTs, it is stipulated that the synchronization delay must not surpass the duration of a single time slot, i.e. $\tau_k[n]<\delta_t$.
$D_k[n]$ and $C_k[n]$ denote the size of DT synchronized data and the average count of central process unit (CPU) cycles needed to process one bit of data for MU $k$, respectively.
When MU $k$ is associated with BS $m^\prime$, i.e., $\alpha_{k,m^\prime}[n]=1$, the DT synchronization delay for MU $k$ is determined by the following calculation:
\begin{equation}\label{eq:tk}
    t_k[n]=\begin{cases}
        \frac{D_k[n]}{R_k[n]}+\frac{D_k[n]C_k[n]}{f_{k,\beta_k(q)}[n]},                                       ~\quad\qquad\text{$\alpha_{k,m^\prime}[n]=\alpha_{k,\beta_k(q)}[n]$},    \\
        \frac{D_k[n]}{R_k[n]}+\frac{D_k[n]}{w_k[n]}+\frac{D_k[n]C_k[n]}{f_{k,\beta_k(q)}[n]},  \text{$\alpha_{k,m^\prime}[n]\neq\alpha_{k,\beta_k(q)}[n]$},
    \end{cases}
\end{equation}
where $f_{k,\beta_k(q)}[n]$ is the computing frequency allocated to MU $k$ by MEC server $\beta_k(q)$. \eqref{eq:tk} indicates that when the BS associated by MU $k$ is the same as the MEC server where the DT of MU is deployed, i.e., $\alpha_{k,m^\prime}[n]=\alpha_{k,\beta_k(q)}[n] $, the DT synchronization delay $t_k[n]$ for consists only of MU $k$'s transmission delay and the processing delay at MEC server $\beta_k(q)$, otherwise the transmission delay between BSs for the MU's data needs to be considered, i.e., $D_k[n]/w_k[n]$. Considering the limited computing capability per MEC server, we have the following constraints
\begin{equation}
	0\leq f_{k,\beta_k(q)}[n]\leq f_{\beta_k(q)}^{\max} ,
\end{equation}
\begin{equation}
	0\leq \sum_{k=1}^{K}f_{k,m}[n]\leq f_{m}^{\max},
\end{equation}
where $f_{m}^{\max}$ is the maximum computing frequency of the MEC server $m$.

\subsection{DT Migration Model}\label{subsec:migr}
The DT migration may occur when $ \beta_k(q)\neq\beta_k(q-1)$. If the DT migration of MU $k$ is decided, $G_k$ consecutive time slots will be utilized for DT migration, where $0\leq G_k<T$, and the collection of time slots allocated for DT migration is defined as  $\mathcal{T}_q^{G_k}\triangleq\{qT,qT+1,\ldots,qT+G_k-1\}$. Corresponding to this, we have $\mathcal{T}_q^{G_k}=\{\emptyset\}$ when the DT migration of the MU is not triggered. Note that the synchronization of MU is not allowed during the migration of the MU's DT.

However, the occurrence of synchronization failure cannot be entirely prevented due to migration delay of DT and wireless channel fluctuations. 
More specifically, when a MU initiates a synchronization request, it may encounter failure under two conditions. Firstly, if its DT is undergoing migration, it may be unable to establish synchronization successfully. Secondly, if the synchronization time required exceeds the specified latency requirement, the synchronization request is deemed unsuccessful.
To quantify this variable, we establish the DT synchronization failure factor $X_k[n]$, which takes a value from the set $\{0,1\}$ for each time interval $n$. If $X_k[n]$ is assigned a value of 1, it signifies that a synchronization attempt has failed at time slot $n$. Conversely, a value of 0 indicates a successful synchronization.
The synchronization failure factor $X_k[n]$ is mathematically articulated as follows:
\begin{equation}
	X_k[n]=\begin{cases}
		\mathbb{I}(a_k[n]=1),                  & \text{if $n\in\mathcal{T}_q^{G_k}$},                         \\
		\mathbb{I}(a_k[n]=1,t_k[n]>\tau_k[n]), & \text{if $n\in\mathcal{T}_q\backslash \mathcal{T}_q^{G_k}$},
	\end{cases}
\end{equation}
where $\mathbb{I}(\cdot)=1$ if DT synchronization fails, and $\mathbb{I}(\cdot)=0$ otherwise.

\subsection{Problem Formulation}
Considering that the MUs usually have limited energy resources that are difficult to replenish, it is advantageous to maintain the DTs at the BS, which have abundant power supply.  In addition, the main energy consumption of MU is caused by transmission power, and optimizing transmission energy consumption is beneficial for maintaining DT synchronization for long time. The energy consumption of MU $k$ can be quantified using the following expression:
\begin{equation}
	E_k[n]=p_k[n]\frac{D_k[n]}{R_{k}[n]}.
\end{equation}
Particularly, we aim at minimizing the long-term energy consumption of the MUs by jointly optimizing the transmission power of MUs ${\bf p}\triangleq\{p_k[n], k\in \mathcal{K},n \in \mathcal{N}\}$, the transmission rate between BSs ${\bf w}\triangleq\{w_{i,j,k}[n], i\neq j, i,j \in \mathcal{M},k \in \mathcal{K},n \in \mathcal{N}\}$, the DT deployment ${\bm \beta}\triangleq\{\beta_k(q), k \in \mathcal{K},  q \in \mathcal{Q}\}$, and the computation resource allocation at MEC servers ${\bf f}\triangleq\{f_{k,m}[n], k \in \mathcal{K}, m \in \mathcal{M},  n \in \mathcal{N}\}$. Building upon the previously stated definitions, the formulation of the optimization problem is as follows:
\begin{subequations}\label{P0}
	\begin{align}
		\mathop {\min } \limits_{{\bm \beta},{\bf p,w,f}} & \lim\limits_{Q\rightarrow\infty} \frac{1}{QKT}  \left( { \sum\limits_{q=0}^{Q-1} \sum\limits_{n\in\mathcal{T}_q} \sum\limits_{k=1}^K \mathbb{E} \{E_k[n]\}}\right) \label{P0:ob} \\
		\text{s.t.}
		& ~ \lim\limits_{Q\rightarrow\infty} \frac{1}{QT} \sum\limits_{q=0}^{Q-1} \sum\limits_{n\in\mathcal{T}_q}  \mathbb{E}\{X_k[n]\}\leq\varepsilon,k \in \mathcal{K},\label{P0:rel}    \\
		& \beta_k(q)\in \mathcal{M}, k \in \mathcal{K}, q \in \mathcal{Q},\label{P0:beta}                                                                                                \\
		& 0 \le p_k[n] \le P_{k}^{{\rm{max}}}, k \in \mathcal{K}, n \in \mathcal{N},\label{P0:p}                                                                                         \\
		& \nonumber0\leq w_{i,j,k}[n]\leq w_{i,j} ,                                                                                                                                      \\&\quad i\neq j, i, j \in \mathcal{M},  k \in \mathcal{K},n \in \mathcal{N},\label{P0:wi}\\
		& \nonumber 0\leq \sum_{k=1}^{K}\zeta_{k,i,j}[n]w_{i,j,k}[n]\leq w_{i,j},                                                                                                                         \\&\quad i\neq j, i,j \in \mathcal{M},n \in \mathcal{N},\label{P0:wis}\\
		& \nonumber0\leq f_{k,\beta_k(q)}[n]\leq f_{\beta_k(q)}^{\max},                                                                                                                  \\&\quad k \in \mathcal{K}, q\in\mathcal{Q},n\in\mathcal{N},\label{P0:fkt}\\
		& 0\leq \sum_{k=1}^{K}f_{k,m}[n]\leq f_{m}^{\max},  m \in \mathcal{M}, n\in\mathcal{N}\label{P0:fk},
	\end{align}
\end{subequations}
where  $\mathbb{E}\{\cdot\} $ represents the expectation, which is taken over the MU mobility, DT synchronization request, and channel randomness. $\varepsilon$ is the threshold cap of the DT synchronization failure ratio. 
$P_{k}^{{\rm{max}}}$ is the maximum transmission power of the MU $k$.
\eqref{P0:rel} limits the long-term failure ratio of DT synchronization.
\eqref{P0:beta} is the association factor between the MU $k$ and the MEC server $m$.
\eqref{P0:p} limits the transmission power of the MU $k$. 
\eqref{P0:wi} and \eqref{P0:wis} limit the wired transmission rate from BS $i$ to BS $j$.
\eqref{P0:fkt} and \eqref{P0:fk} are the allocation of computation resources for the MEC server.

That problem \eqref{P0} is a nonlinear integer programming problem due to the nonconvexity of objective function and the existence of integer variables.
This problem becomes even more complex to capture the real-time decision-making mechanism since it involves the mobility of MUs, DT synchronization requests, and channel randomness. Specifically, in the absence of prior knowledge about the dynamic channel conditions and the network environment, traditional offline algorithms struggle with rendering real-time decisions to arrive at a solution for the problem. This is because the typical offline optimization algorithm needs to know all the state information of the network before solving the optimization problem. Therefore, traditional iterative offline algorithms make it hard to solve the problem timely. 
As an ML method, MADRL is capable of interacting and learning from the environment and finally obtains a policy model that can be deployed on the devices, thereby facilitating real-time decisions and meeting long-term benefits according to the current state.

\section{Problem Transformation and POMDP Formulation} \label{s:mdp}
To address the long-term reliability constraint present in problem \eqref{P0}, we will first reformulate it utilizing the  Lyapunov optimization theory.
Next, we formulate the problem as a POMDP with multiple agents to make it suitable for MADRL.

\subsection{Problem Transformation}
Implementing conventional MADRL techniques to the problem presented is complex due to the complications involved in addressing the long-term reliability constraint \eqref{P0:rel} highlighted in \cite{Wu2021TII, Xu2022JSAC}.
Therefore, we utilize the virtual queue method combined with Lyapunov optimization theory to reformulate our problem \cite{neely2010stochastic}. Therefore, to satisfy the reliability constraint \eqref{P0:rel}, we create a virtual queue vector ${\bf Y}[n]=\left[Y_1[n], Y_2[n],\dots, Y_k[n]\right]^{\rm T}$ that models the DT synchronization reliability conditions for each MU. The variable $Y_k[n]$ changes according to the following rule:
\begin{equation}
	Y_k[n+1]=\left[Y_k[n]+X_k[n]- \varepsilon \right]^{+}, \label{eq:que}
\end{equation}
where $[\cdot]^{+}=\max\{\cdot,0\}$. $Y_k[n]$ refers to the queue length at the time slot $n$, initialized as $Y_k[0]=0$, which measures the extent by which the current backlog of synchronization failures surpasses the threshold $\varepsilon$.
Utilizing Lyapunov optimization theory, the reliability constraint \eqref{P0:rel} can be transformed into a mean-rate stability condition for the virtual queue, such that $\lim \limits_{n\rightarrow\infty} {\mathbb E}\{Y_k[n]\}/n=0$.

Then, we define the Lyapunov function $\mathcal{L}({\bf Y}[n]) = \frac{1}{2}{\bf Y}^{\rm T}[n]{\bf Y}[n]$ to assess how well the system adheres to the reliability constraints. 
A low value of this function indicates a reliable and stable system, while a high value suggests potential reliability issues, requiring corrective measures for improved performance.
Moreover, the conditional Lyapunov drift over a $T$-slot duration, also referred to as a frame, is denoted by $\Delta \mathcal{L}[n]=\mathbb{E}\left\{\mathcal{L}({\bf Y}[n+T])-\mathcal{L}({\bf Y}[n])\vert {\bf Y}[n] \right\}$, and the upper bounds for this are provided by the two lemmas presented below.
\begin{lemma}
	Define $\Delta \mathcal{L}_k[n]=\mathbb{E}\left\{\mathcal{L}(Y_k[n+T])-\mathcal{L}(Y_k[n])\vert Y_k[n] \right\}$, then $\Delta \mathcal{L}_k[n]$ can be upper bounded by
	\begin{equation} \label{eq:up1}
		\Delta \mathcal{L}_k[n]\leq B_{k,1}T+\mathbb{E}\left\{\sum \limits_{n\in\mathcal{T}_q}Y_k[n]\left[X_k[n]-\varepsilon \right]\bigg\vert Y_k[qT] \right\},
	\end{equation}
	where $B_{k,1}\triangleq \frac{1}{2}\left(\lambda_k+\varepsilon^2\right)$ is a constant.
\end{lemma}
\begin{proof}
See Appendix A.
\end{proof}

The above lemma shows that the Lyapunov drift $\Delta \mathcal{L}_k[n]$ is upper bounded by a constant $B_{k,1}$ plus a term related to the queue length $Y_k[n]$ and the number of synchronization failures $X_k[n]$. However, it is difficult to directly minimize the term $\mathbb{E}\{\sum \limits_{n\in\mathcal{T}_q} (Y_k[n])[X_k[n]-\varepsilon]\mid Y_k[qT]\}$ at the beginning of each time frame. This is because the future values of $Y_k[n]$ over $n\in \{qT+1,\ldots,(q+1)T-1\}$ are unknown during the time slot $n=qT$ and are hard to be predicted due to the uncertain accumulation of synchronization failed in the queue \cite{Yao2014TPDS,Liang2022TWC}. In response to this matter, an upper bound for this expression will be set forth in the ensuing lemma.
\begin{lemma}
$\Delta \mathcal{L}_k[n]$ can be upper bounded by
	\begin{equation} \label{eq:up2}
		\Delta \mathcal{L}_k[n]\leq B_{k,2}T+\mathbb{E}\left\{Y_k[qT]\sum \limits_{n\in\mathcal{T}_q}\left[X_k[n]-\varepsilon \right]\bigg\vert Y_k[qT] \right\},
	\end{equation}
	where $B_{k,2}\triangleq B_{k,1}+(T-1)[(1-\varepsilon)\lambda_k+\varepsilon^2]/2$ is a constant.
\end{lemma}

\begin{proof}
See Appendix B.
\end{proof}
Hence, Lemma 2 provides an upper bound for the drift-plus-penalty function  $\Delta \mathcal{L}_k[n]$, eliminating the requirement of predicting the future value of $Y_k[n]$ within each time frame. This upper bound significantly reduces the complexity of the problem and makes it well-suited for the two-timescale model.

Based on the Lyapunov optimization theory, the objective function of the original problem \eqref{P0}, which aims to minimize the long-term average energy consumption while ensuring the DT synchronization reliability constraints, can be reformulated as the minimization of the drift-minus-bounds function, expressed as follows:
\begin{equation}
	\begin{split}
		&\eta\Delta \mathcal{L}[n]+\frac{1}{QKT}  \left( { \sum\limits_{q=0}^{Q-1} \sum\limits_{n\in\mathcal{T}_q} \sum\limits_{k=1}^K \mathbb{E} \{E_k[n]\}}\right)\\
		&\leq\eta\sum\limits_{q=0}^{Q-1} \sum\limits_{n\in\mathcal{T}_q} \sum\limits_{k=1}^K\left\{ B_{k,2}T+\mathbb{E}\left\{Y_k[qT]\sum \limits_{n\in\mathcal{T}_q}\left[X_k[n]-\varepsilon \right]\right\}\right\}\\
		&+\sum\limits_{q=0}^{Q-1} \sum\limits_{n\in\mathcal{T}_q} \sum\limits_{k=1}^K\mathbb{E}\left\{\frac{E_k[n]}{QKT} \right\},
	\end{split}
\end{equation}
where $\eta$ denotes the control factor that governs the tradeoff between energy consumption and DT synchronization reliability. 
By disregarding terms in the objective function that are unrelated to the optimization variables, problem \eqref{P0} can be rewritten as follows:
\begin{subequations}\label{P1}
	\begin{align}
		\mathop {\min } \limits_{{\bm \beta},{\bf p,w,f}} & \lim\limits_{Q\rightarrow\infty} \sum\limits_{q=0}^{Q-1} \sum\limits_{n\in\mathcal{T}_q} \sum\limits_{k=1}^K\mathbb{E}\left\{\Xi_k[n]\right\} \label{P1:ob} \\
		\text{s.t.}
		& ~ \text{(15c) -- (15h)},
	\end{align}
\end{subequations}
where $\Xi_k[n]=\frac{E_k[n]}{QKT}+\eta Y_k[qT][X_k[n]-\varepsilon]$.

\subsection{Modeling of the POMDP}
Traditional optimization approaches typically assume complete information, making them poorly suited for dynamic and uncertain environments.
Consequently, we reformulate the problem \eqref{P1} as a POMDP to enable informed sequential decision-making that can dynamically adapt to the evolving environment and its inherent uncertainties.
By expressing our problem as a POMDP, we can leverage a distributed framework to efficiently and effectively address the curse of dimensionality arising from the involvement of multiple BSs and MUs in the network, while keeping the computational complexity low. Hereinafter, we first introduce the fundamental components of a POMDP, and then we discuss each component in the context of our problem.

A POMDP of an agent set $\mathcal{I}$ can be denoted as $<\mathcal{I},\mathcal{O},\mathcal{S},\mathcal{A},\mathcal{P},\mathcal{R},\pi>$, which is composed of state space $\mathcal{S}$, observation space $\mathcal{O}$, action space $\mathcal{A}$, reward function $\mathcal{R}$, probability of environment transferring $\mathcal{P}$, and policy $\pi$. We consider one time slot as a single time step.  The time slots and time steps are used alternately in the following paper. 
The global state $s[n]\in\mathcal{S}$ is partially observable to agents, especially in privacy-aware MEC networks and distributed structures. Consequently, the agents can only get an observation $o_i[n]\in\mathcal{O}$ from the environment. 
$\pi_i(a_i[n]\vert o_i[n])$ denotes the probability of taking action $a_i[n]$ at state $o_i[n]$,
and then the reward $r[n]$ is given by function $\mathcal{R}$ after the action.
In our model, the agents include each BS, MU, and the control center. On this basis, the total number of agents, denoted by $I$, is given by $I=K+M+1$.
The partial observation of agents $o_i[n]$, where $i\in\mathcal{I}=\{1,2,\ldots, I\}$, at time step $n$ can be merged to construct the global state of environment $s[n]$. It can be observed that the MUs, BSs, and the control center can be classified into three types of agents. We denote the set of MU agents as $I_u=\{1,2,\ldots,K\}$ and the set of BS agents as $I_U=\{K+1,K+2,\ldots,K+M\}$. The POMDP elements are formulated as follows:
\subsubsection{MU agent} Each MU acts as an agent, collecting its own synchronization information and location. Given that MUs often have limited energy resources that are difficult to replenish, it is crucial to manage their power consumption carefully. Based on MU agents observations, the MU agents determine their associated BS, and decide on the transmission power according to their policy. The primary energy consumption of MUs is due to transmission power; thus, optimizing this energy expenditure is essential for maintaining long-term DT synchronization and enhancing overall network efficiency. The observation and action of MU agents are illustrated below.

\textbf{Observation}: MUs can obtain their own DT synchronization information, and the coordinates of themselves and BSs via location service, the observation $o_k[n],  k \in \mathcal{K}$ can be defined by
\begin{equation}
	o_k[n]=\left\{k, {\bf u}_k[n],{\bf q}_m[n], \Omega_k[n],  m \in \mathcal{M}\right\}.
\end{equation}
Note that we add the index of MUs to distinguish their observations.

\textbf{Action}: After receiving the current environmental information, the MU agents need to associate with the BS, and allocate the transmission power accordingly. The action of MU $k$ can be represented by
\begin{equation}
	a_k[n]=\left\{ \alpha_{k,m}[n], p_k[n], m \in \mathcal{M}\right\}.
\end{equation}

\subsubsection{BS agent} We define the set of MUs requesting DT synchronization as $\mathcal{K}^{\star} = \{k \in \mathcal{K} \mid a_k[n]=1\}$.
The BS can obtain the current location information of the MUs using positioning services.
After the BSs receive the DT synchronization requests from the MUs, the BSs initiate the synchronization process.
Note that, the synchronization request $\Omega_k[n]$ of the MU only consists of a triplet, and its size is very small. 
Therefore, the data exchange of the MU's DT synchronization requests between BSs can be considered negligible.
Subsequently, they allocate both communication and computation resources for MUs.  Thus, the BSs can be treated as agents, whose MDP elements are defined as follows:

\textbf{Observation}: BS $m$ can observe the DT synchronization information and the location of MUs by exchanging the information of its served MUs with other BSs. As a result, the observation of BS $m$ is given by
\begin{equation}
	\begin{split}
		o_{K+m}[n]=&\left\{ m, \alpha_{k,m}[n], {\bf u}_k[n], {\bf q}_{m}[n],{\bf q}_{-m}[n], \Omega_k[n],  \right.\\
		& \left. \zeta_{k,m,m^{\prime}}[n], k \in\mathcal{K}^{\star},m^{\prime} \in \mathcal{M} \right\},
	\end{split}
\end{equation}
where $-m$ denotes the set in $\mathcal{M}$ excluding $m$.

\textbf{Action}: After receiving the current environment information, the BS agents need to allocate the wired transmission rate and computation resource allocation according to the observations. The actions of BS $m$ can be denoted by
\begin{equation}
	a_{K+m}[n]=\left\{ w_{m,j,k}[n], f_{m,k}[n], k \in \mathcal{K}^{\star}, j \in \mathcal{M} \right\}.
\end{equation}

\subsubsection{Control center agent} At the end of each time frame, the control center needs to observe the locations of BSs and MUs, and the association status. Then, it decides the DT deployment variable ${\bm \beta}$ for DT migration. Typically, the central controller is equipped with powerful cloud servers. In addition, transmission between the BSs and the central controller is conducted via wired communication. Therefore, the processing delay and communication delay of the central controller can be disregarded.

\textbf{Observation}: The observation of the control center agent is defined as follows:
\begin{equation}
	o_{K+M+1}[n]=\left\{ \alpha_{k,m}[n], {\bf q}_m[n], {\bf u}_k[n], m \in \mathcal{M}, k \in \mathcal{K} \right\}.
\end{equation}

\textbf{Action}: The control center agent manages the deployment of the DT, and the corresponding actions are defined accordingly as:
\begin{equation}
	a_{K+M+1}[n]=\left\{ \beta_k(q), k \in \mathcal{K} \right\}.
\end{equation}

\textbf{Reward:} The sampling efficiency of the control center agent will be much lower than other agents if it only makes one step during each time frame. This creates inconsistency between the learning speed of agents, and thus sharply increases the time consumption and complexity of training. To relieve the problem, we allow the control center agent to make actions not only in time slot $n=qT$, but also in other time slots during time frame $q$. 
We define an independent action $\tilde{a}_{K+M+1}[n]=\{ \tilde{\beta}_k[n],  k \in \mathcal{K}, n \in \mathcal{N}\}$, which do not affect the real DT deployment $\beta_k(q)$ from $qT+1$-th to $(q+1)T$-th time slot.
Hence, the independent actions are evaluated by a reward that estimates the total time consumption of the MUs imposed by transferring and computing by the BSs. The reward function takes the following form:
\begin{equation}\label{eq:cenre}
	r_c[n]=-\frac{1}{K}\sum\limits_{k=1}^K r_{c,k}[n],
\end{equation}
where
\begin{equation}\label{eq:cenrek}
	r_{c,k}[n]= 	\begin{cases}
		\frac{D_k[n]C_k[n]}{f_{k,\tilde{\beta}_k[n]}[n]}, & \text{$\alpha_{k,m}[n]=\alpha_{k,\tilde{\beta}_k[n]}[n]$},    \\
		\frac{D_k[n]}{w_k[n]}+\frac{D_k[n]C_k[n]}{f_{k,\tilde{\beta}_k[n]}[n]}, & \text{$\alpha_{k,m}[n]\neq\alpha_{k,\tilde{\beta}_k[n]}[n]$}.
	\end{cases}
\end{equation}
In other words, we only set $\beta_k(q)=\tilde{\beta}_k[q\cdot T]$ in time slot $n=qT$, and collect rewards $r_c[n]$ during each time slot $n\in\mathcal{T}_q$.

\subsubsection{Global reward} According to the objective function formulated by problem \eqref{P1}, the global reward function of the network is defined as follows for all MU agents and BS agents:
\begin{align}\label{eq:reward}
	r_g[n]=-\sum\limits_{k=1}^K\Xi_k[n].
\end{align}
The global reward can be considered as the reward of MUs and BSs.
For notational simplicity, we represent the reward for any one of the agents, including $r_g[n]$ and $r_c[n]$, as $r[n]$.

Due to the dynamic network environments and the difficulty in precisely obtaining the probability of environment transferring $\mathcal{P}$ for POMDP, obtaining an exact POMDP model becomes challenging. Consequently, traditional solution methods such as dynamic programming become impractical. As a result, MADRL, which is capable of operating without a known model and learning from environmental interactions, emerges as the preferred approach for solving the POMDP problem.

\section{DT Synchronization BETA-HAPPO Algorithm} \label{s:proposed}
After modeling the proposed system as a POMDP, we tackle the problem using MADRL methods that orchestrate multiple agents to cooperatively learn the optimal policy through environmental interactions. Compared to single-agent DRL, MADRL offers greater extensibility for high-dimensional action spaces and better adaptation to distributed systems. The centralized training and decentralized execution (CTDE) framework, a notable paradigm in MADRL, enables agents to access the global state during training and execute actions independently \cite{lowe2017multi}, thus mitigating the negative impacts of environmental variability. Within this framework, the HAPPO algorithm represents a pivotal advancement \cite{Jakub2022ICLR_Trust}. The HAPPO algorithm is a policy-based trust region learning algorithm that employs the widely used CTDE mechanism,  for more effective collaborative learning. The HAPPO algorithm does not require agents of the same type to share identical neural network parameters for their distributed actor and critic networks. This approach prevents the exponential degradation of the reward function that usually happens as the number of agents increases.

\begin{figure*}[t]
	\centering
	\includegraphics[width=\textwidth]{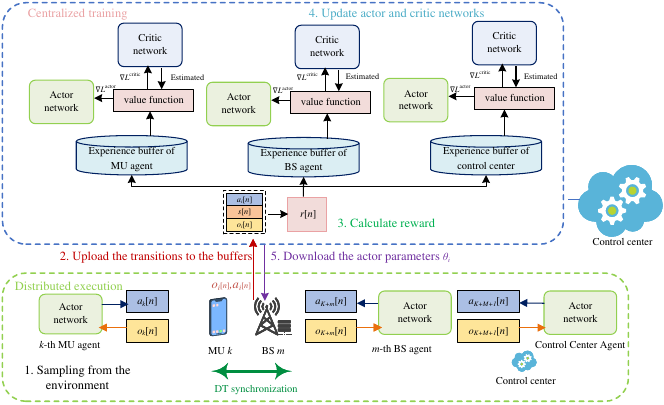}
 \vspace{-2em}
	\caption{An illustration of the Beta-HAPPO-based training framework.} \label{fig:framework}
	
     \vspace{-1.5em}
\end{figure*}

We propose a training framework based on HAPPO, where each MU, BS, and the control center are represented as individual agents. Each agent possesses an actor network and a critic network, as illustrated in Fig. \ref{fig:framework}.    
${\bm\theta}_i$ and ${\bm\xi}_i$ denote the actor network parameters and the critic network parameters for each agent $i$, respectively.
Their joint policy is represented as ${\bm \pi}$.
We regard a time slot as a decision time step in the POMDP.
Our proposed training architecture is divided into two key phases: environmental interaction and policy updating. 

During the environmental interaction phase, each agent observes the partial environment $o_i$, feeding this data into both the actor and critic networks.  
Following this, the MUs, BSs, and the control center utilize the actor network to generate actions $a_i$ based on partial observations $o_i$ according to the policy function $\pi_i(a_i\vert o_i)$. Subsequently, the control center merges these partial observations $o_i$ into the global state $s$. Moreover, the control center combines the global state $s$, the actions $a_i$, and other environmental details to compute the reward $r$, which is then stored in the control center's experience buffer.
During the policy update phase, a selection of samples is drawn from the experience buffer to train the actor and critic networks. 
Meanwhile, the critic networks, informed by a global environmental state $s$, evaluate the joint policy $\pi$ and assist the actor network in refining its policy.
We adopt a gradient-based approach to update the actor and critic network parameters. Let $L_{{\rm actor}}({\bm\theta}_i)$ and $L_{{\rm critic}}({\bm\xi}_i)$ represent the loss function of the actor and critic network of agent $i$, respectively. The derivation of the loss functions will be presented in the subsequent part.

Before defining the loss function of the actor network, it is crucial to introduce the concept of the multi-agent advantage function, as it serves as a key component in optimizing the actor network's loss function.
We first define two disjoint and ordered subsets of agents, denoted as $i_{1:u} = \{i_1, i_2, \ldots, i_u\}$ and $j_{1:v} = \{j_1, \ldots, j_v\}$. Subsequently, given the actions $\mathbf{a}_{j_{1:v}}$ taken by agents $j_{1:v}$, the multi-agent advantage function $A_{i_{1:u}}(s,{\bf a}_{j_{1:v}},{\bf a}_{i_{1:u}})$ quantifies the expected performance gain of agents $i_{1:u}$ when they take actions $\mathbf{a}_{i_{1:u}}$.
Specifically, if $A_{i_{1:u}}(s, \mathbf{a}_{j_{1:v}}, \mathbf{a}_{i_{1:u}}) > 0$, then the actions $\mathbf{a}_{i_{1:u}}$ perform better than what the average or baseline policy would predict. In other words, the execution of this action results in a reward that is higher than the average expectation, thus incentivizing agents \( i_{1:u} \) to take this action.
Otherwise, the action taken by the agents leads to a reward that is below the average level, suggesting that the agents should typically avoid this action. 
Furthermore, we introduce the multi-agent advantage function decomposition lemma in \cite{Jakub2022ICLR_Trust}, which is to quantify and analyze the impact of the actions taken by a subset of agents on the overall performance.
To simplify the notation, we omit subscript $i$ and elevate the corresponding subscripts to replace the original position of $i$.
\begin{lemma}
	(Multi-agent advantage function decomposition lemma\cite{Jakub2022ICLR_Trust}) For a given joint policy $\bm{\pi}$, any state $s$, any joint action ${\bf a}_{{1:u}}$ for agent subset $i_{1:u}$, the advantage function $A_{{1:u}}(s,{\bf a}_{{1:u-1}})$ can be decomposed as follows:
	\begin{equation}
		A_{{1:u}}(s,{\bf a}_{{1:u}})=\sum_{j=1}^u A_{j}(s,{\bf a}_{{1:j-1}},a_{j}),
	\end{equation}
where $A_{j}(s,{\bf a}_{{1:j-1}},a_{j})$ quantifies the expected performance gain of agent $i_{j}$ when it takes actions $a_{{j}}$, given the actions $\mathbf{a}_{{1:j-1}}$ taken by agents $i_{1:j-1}$.
\end{lemma}

Correspondingly, Lemma 3 facilitates the use of a sequential policy update scheme. By decomposing the global advantage function ${A_{{1:u}}(s,{\bf a}_{{1:u}})}$ into local advantage functions ${A_{j}(s,{\bf a}_{{1:j-1}},a_{j})}$ for each agent and utilizing a sequential policy update scheme, this approach allows agents to monotonically improve their policies based on their advantage functions. 
To ensure this improvement, it is imperative to satisfy the condition ${A_{j}(s,{\bf a}_{{1:j-1}},a_{j})\geq 0}$. By continuously enhancing its policy in a monotonic manner, the agent can meet this condition and thereby achieve a continuous improvement in the overall system performance within a multi-agent setting.
Building on this, a sequential update scheme is adopted, which randomly draws a permutation $u=\{u_1,\ldots,u_I\}$ and updates the policy of agents $i_{1:I}$ by the order $\{i_{u_1},\ldots,i_{u_I}\}$.
For any given joint policy ${\bm \pi}$, action ${{\bf a}_{u}}$ of agent ${i_u}$ and joint action ${\bf a}$, the agents update their actor network according to the surrogate objective function defined as follows\cite{Jakub2022ICLR_Trust}:
\begin{align}\label{eq:L_actor}
	L_{{\rm actor}} & (\bm\theta_{u})=\mathbb{E}\Bigg\{ \min\Bigg[\frac{\pi_{u}(a_{u}[n]\vert o_{u}[n];\bm\theta_{u})}{\pi_{u}(a_{u}[n]\vert o_{u}[n];\tilde{\bm\theta}_{u})} M_{{1:u}}(s,{\bf a})[n], \notag\\
	&{\rm clip}\left( \frac{\pi_{u}(a_{u}[n]\vert o_{u}[n];\bm\theta_{u})}{\pi_{u}(a_u[n]\vert o_{u}[n];\tilde{\bm\theta}_{u})},1\pm\epsilon \right) M_{{1:u}}(s,{\bf a})[n]  \Bigg]  \notag\\
	&+\psi S(o_{u}[n]) \Bigg\},
\end{align}
where {$\tilde{\bm \theta}_{u}$} represents the actor network parameters of agent $i_u$ following the most recent update. $\psi$ is a hyperparameter used to adjust the intensity of exploration.
The first term in \eqref{eq:L_actor} is the clipped surrogate objective function, which alleviates the policy from changing too much as well as performance fluctuation, and the second term $S(o_{u}[n])$ is the entropy regularization term to enhance the ability for exploration\cite{Cai2022TNSE}.
{$M_{{1:u}}(s,{\bf a})[n]=\frac{{\bm \pi}_{{1:u-1}}({\bf a}_{{1:u-1}}[n]\vert o_{1:u-1}[n];\bm\theta_{u})}{{\bm \pi}_{{1:u-1}}({\bf a}_{{1:u-1}}[n]\vert o_{1:u-1}[n];\tilde{\bm\theta}_{u})}\hat{A}_{u}(s[n])$} is calculated to estimate the mathematical expectation of $A_{u}(s,{\bf a}_{{1:u-1}},a_{u})$,
where {${\bm \pi}_{{1:u-1}}(\cdot)$} is the joint policy of agent subset $i_{1:u-1}$, and {$\hat{A}_{u}(s[n])$} is the agent $i_u$ estimated advantage function based on generalized advantage estimation (GAE)\cite{Schulman2015arxiV}. 
{For any given global state $s[n]$ and reward $r[n]$ at time slot $n$, the formulation of the GAE function is as follows:
\begin{align} \label{eq:GAE}
	\hat{A}_{u}(s[n])=  &\sum\limits_{l=0}^\infty(\gamma \varrho)^l\Big(r[n+l]+\gamma V_{\bm\xi_{u}}\big(s[n+1+l]\big) \notag 
	               \\ &-V_{\bm\xi_{u}}(s[n+l])\Big),
\end{align}
where $ \varrho$ reveals the tradeoff status between variance and bias in training.}
$\gamma$ serves as the discount factor, signifying the weight assigned to forthcoming rewards by the agents.
$V_{\bm\xi_{u}}(s[n])$ indicates the value-function approximated with the critic network. Here, the critic network parameter $\bm\xi_{u}$ of agent $i_u$  is updated through following loss function minimization\cite{Grondman2012TSMCC}:
\begin{equation}\label{eq:L_critic}
	L_{{\rm critic}}(\xi_{u})=\mathbb{E}\left\{ \frac{1}{2}\left(r[n]+\gamma V_{\xi_{u}}\big(s[n+1])  
	-V_{\xi_{u}}(s[n])\right)^2 \right\}.
\end{equation}
Thus, the actor networks are improved through gradient ascent, while the critic networks are optimized using gradient descent \cite{Grondman2012TSMCC}. The formulas for these optimizations are expressed as follows:
\begin{equation}
	\bm\theta_{u}\leftarrow\bm\theta_{u}+l_a\nabla L_{\rm actor}(\bm\theta_{u}), \label{eq:updateactor}
\end{equation}
\begin{equation}
	\bm \xi_{u}\leftarrow\bm \xi_{u}-l_c\nabla L_{\rm critic}(\bm \xi_{u}), \label{eq:updatecritic}
\end{equation}
where $l_a$ and $l_c$ are learning rate.

The aforementioned actions are typically subject to continuous and bounded constraints, such as \eqref{P0:p} and \eqref{P0:fkt}. Nevertheless, the conventional approach of sampling actions from a Gaussian distribution in policy networks can lead to a skewed estimation of the policy gradient, due to the distortion caused by clipping the values for actions that fall outside the permissible range. To address this issue, we propose to use Beta distribution instead of Gaussian distribution to model the output of policy networks. {The Beta distribution is characterized by the following expression:
\begin{equation}
	f(x,\vartheta,\chi)=\frac{\Gamma(\vartheta+\chi)}{\Gamma(\vartheta)\Gamma(\chi)}x^{\vartheta-1}\Big(1-x\Big)^{\chi-1}, \label{eq:Beta}
\end{equation}
where the shape of the distribution is governed by the parameters $\vartheta$ and $\chi$. }
As the domain of the Beta distribution, denoted by \eqref{eq:Beta}, is restricted, it is well-suited for the sampling of actions that are constrained within certain limits.
Furthermore, Beta distribution can be initialized to be more uniform than Gaussian distribution, enabling agents to explore the action space more extensively during the initial stage of training.

Based on the discussions above, we introduce the Beta-HAPPO training framework, summarized in Algorithm \ref{alg:happo-framework}. The framework comprises a main loop, which unfolds in two primary stages. 
Lines 2-16 introduce the environmental interaction phase, which encompasses the agent's interaction with the environment using the current policy. The outcomes of these interactions are then stored in an experience buffer to be utilized for future network parameter updates.
Subsequently, GAE is calculated in Line 17. 
Lines 18-27 refer to the policy updating phase, which involves updating the network parameters using samples drawn from these environment interactions.
It updates the actor and critic networks by sampling mini-batches from the experience buffer and then clears the buffer for the next iteration. 
With numerous iterations, the policy is refined and redeployed for environmental interaction, thus gathering new samples.

{ Algorithm 1 can be computed in parallel, leading to increased efficiency compared to the serial way.
We utilize a multi-layer perceptron (MLP) in our neural network architecture. In the MLP, the computational complexity of the $j$-th layer is given by $\mathcal{O}(U_{j-1}U_{j}+U_{j}U_{j+1})$, where $U_j$ represents the number of neurons in the $j$-th layer. Consequently, the overall computational complexity of a $J$-layer MLP can be calculated as $\mathcal{O}(\sum_{j=2}^{J-1} U_{j-1}U_{j}+U_{j}U_{j+1})$.
Since the actor and critic networks are based on MLP, the overall computational complexity for ${\rm Me}$ episodes, which have the length of $Q\cdot T$ steps and update times of each training step ${\rm Rp}$, is calculated by $\mathcal{O}\left({\rm Me}(({\rm Rp}+Q\cdot T)\sum_{j=2}^{J-1} U_{j-1}U_{j}+U_{j}U_{j+1}) \right)$. Note that in the one-step decision, the same type of agents execute actions simultaneously. The computational complexity is with respect to one forward propagation by actor networks, which is given by $\mathcal{O}\left(\sum_{j=2}^{J-1} U_{j-1}U_{j}+U_{j}U_{j+1} \right)$.}

\begin{algorithm}[t]
	\caption{Beta-HAPPO-Based MADRL Algorithm}
	\label{alg:happo-framework}
	\begin{algorithmic}[1]
		\STATE{Initialize update times of each training step $\rm{Rp}$, and maximum episodes ${\rm Me}$.}
		\FOR {Episode = 1,$\dots$, ${\rm Me}$}
		\FOR {$q$ = 0,$\dots$, ${\rm Q}-1$}
		\FOR {$n$ = $qT$,$\dots$, $(q+1)T$}
		\STATE {MU agents obtain observation $o_k[n], k \in \mathcal{K}$ from the environment;}
		\STATE {MU agents sample actions $a_k[n]$, $ k \in \mathcal{K}$ according to the actor networks;}
		\STATE {BS agents obtain observation $o_{K+m}[n], m \in \mathcal{M}$ from the environment;}
		\STATE {BS agents sample actions $a_{K+m}[n], m \in \mathcal{M}$ according to the actor networks;}
  	\STATE {Control center agent obtain global state $s[n]$ from the environment;}
        \STATE {Control center samples actions $\tilde{a}_{K+m+1}[n]$ according to the actor networks;}
		\STATE {Upload the transitions to the experience buffers in the control center;}
		\STATE {Control center evaluates the reward according to \eqref{eq:cenre} and \eqref{eq:reward};}
        \STATE {Update the policy entropy and log-probability in experience buffer;}
		\ENDFOR
        \STATE {Control center decides the DT deployment ${\bm \beta}$;}
		\ENDFOR
		\STATE {Calculate the advantage function \eqref{eq:GAE} using GAE;}
		\FOR {epoch = $1,\dots {\rm Rp}$}
		\STATE Randomly generate a permutation $i_{1:I}$ of agents;
		\FOR {agents $u\in\mathcal{I}$}
  	\STATE {Update critic network parameters $\bm\xi_{u}$ according to 
        \eqref{eq:updatecritic};}
		\STATE {Compute $M_{{1:u+1}}(s,{\bf a})=\frac{{\bm \pi}_{u}({\bf a}_{u}\vert o_{u};\bm \theta_{u})}{{\bm \pi}_{u}({\bf a}_{u}\vert o_{u};\tilde{\bm \theta}_{u})}M_{{1:u}}(s,{\bf a})$;}
		\STATE {Update actor network parameters $\bm \theta_{u}$ according to \eqref{eq:updateactor};}
		\ENDFOR
		\ENDFOR
		\STATE {Clear experience buffer;}
		\ENDFOR
		\STATE {Output actor network parameters $\bm \theta_i$ and critic network parameters $\bm \xi_i$.}
	\end{algorithmic}
\end{algorithm}

\section{Numerical Results}\label{s:simulation}
In this section, we evaluate the performance of our proposed Beta-HAPPO scheme by numerical simulations. We adopt the following MADRL benchmark schemes:
\begin{itemize}
	\item \textit{Gaussian-HAPPO}: In this strategy, HAPPO adopts Gaussian distribution in actor networks, i.e., the original version of HAPPO \cite{Jakub2022ICLR_Trust}.
	\item \textit{Beta-MAPPO}: In this scheme, MAPPO adopts Beta distribution in actor networks. Be noted that unlike HAPPO, MAPPO does not adopt a sequential update strategy\cite{Yu2021ar_Surprising}.
	\item \textit{MADDPG}: This strategy is the multi-agent deep deterministic policy gradient (MADDPG) method, which is a policy-based MADRL algorithm. This scheme employs a deterministic policy with exploration noise and does not utilize distribution as the actor output. In contrast to the on-policy PPO-based schemes, MADDPG is an off-policy scheme that relies on a huge replay buffer and the quality of data in it\cite{lowe2017multi}.
\end{itemize}

\subsection{Simulation Settings}
The simulation settings are presented as follows. We consider a 1000 m $\times$ 1000 m square MEC network area, where the MUs are uniformly distributed. 
Unless otherwise specified, we set the number of MUs $K=30$ and the number of BSs $M=5$. 
The average CPU cycles per bit is uniformly generated in $[550,700]$ cycles/bit, where we denote $C_{\max}=700$, and the synchronized data size $D$ is generated with uniform values ranging from 15 to 25 Kb. 
To normalize the reward to [-1,0] for the update, we multiply $r[n]$ by factor $\nu=100$.
Other default settings for environment and algorithm are summarized in Table \ref{tab:env_settings}, 
according to prior works \cite{Liang2022TWC, Waqar2022TITS, Liu2023TWCU}.
{All numerical experiments are conducted using Python 3.7 with PyTorch on an Intel Core i9-13900HX laptop with 32GB of RAM, without GPU involvement.}

\begin{table}
	\caption{Environment Settings and Algorithm Settings}
	\centering
	\renewcommand{\arraystretch}{1.2}
	\begin{tabular}{|p{2.5in}<{\centering}|p{0.6in}<{\centering}|}
		\hline
		\textbf{Parameters} & \textbf{Value}\\
		\hline
		Time frame length $T$ & 100 \\
		\hline
		Frame size $Q$ & 50 \\
		\hline
		Length of time slot $\delta_t$ & 0.05 s \\
		\hline
		Bandwidth $B$ & 10 MHz \\
		\hline
		Total available transmission rate between BSs $w_{i,j}$ & 10 Mb/s \\
		\hline
		Maximum transmit power of MUs $P_k^{\max}$ & 0.5 W \\
		\hline
		Maximum CPU frequency of BSs $f_m^{\max}$ & 10 GHz \\
		\hline
		Probability of request by MUs $\lambda_k$ &0.5 \\
		\hline
		Referenced channel power gain $\rho_0$ &-30 dB\\
		\hline
		Rician factor $\kappa $ & 10\\
		\hline
		Noise power $ \sigma^2$ & -90 dBW\\
		\hline
		Average moving speed of MUs $\bar{s}$ & [2, 10] m/s\\
		\hline
		Width of service area $W$ & 1000 m\\ 
		\hline
		Maximum training episodes $\rm{Me}$ & 150 \\
		\hline
		Discount factor $\gamma$ & 0.9\\
		\hline
		Learning rate $l_a$, $l_c$& 0.0001 \\
		\hline
		Number of mini-batches & 1 \\
		\hline
		Number of each training step $\rm{Rp}$ & 25 \\
		\hline
		Optimizer & Adam \\
		\hline
		Hidden sizes & [128, 64]\\
		\hline
		Variance of exploration noise for MADDPG & 0.5 \\ 
		\hline
		Threshold of DT failure ratio $\varepsilon$ & 0.2 \\
		\hline
		Control factor $\eta$ & 1.0\\
		\hline
	\end{tabular}
	\label{tab:env_settings}
\end{table}

\subsection{Performance Analysis}
We first evaluate the convergence performance of our proposed Beta-HAPPO scheme compared with other benchmarks.  
Figs. \ref{fig:con}, \ref{fig:cenc}, and \ref{fig:obj} show the average reward of MUs and BSs, average reward of control center, and average energy consumption versus time steps, respectively. 
We can observe from Figs. \ref{fig:con} and \ref{fig:cenc} that the rewards of all the schemes gradually increase during training, which reveals the effectiveness of MADRL algorithms. Among the schemes, the proposed Beta-HAPPO scheme has the highest reward and more stable convergence performance. 
Beta-HAPPO is better than Beta-MAPPO, due to the monotone improvement strategy introduced by HAPPO.
Furthermore, the convergence speed of HAPPO-based schemes significantly outperform the MADDPG scheme. It is because the deterministic policy of MADDPG results in weak explorational ability, and the quality of experiences in the buffers is tortuously improved. Moreover, Gaussian-HAPPO reaches a lower reward compared to Beta-HAPPO, validating that the Beta distribution is better than Gaussian distribution in the scenarios with bounded actions.
We further evaluate the average energy consumption of MUs $E_k[n]$ during training in Fig. \ref{fig:obj}. We can observe that the average energy consumption steadily decreases as the training steps increase, indicating the effectiveness of MADRL schemes.
The energy consumption of the proposed Beta-HAPPO is reduced by 20.28\%, 29.52\%, and 41.08\% compared with Beta-MAPPO, Gaussian-HAPPO, and MADDPG, respectively, validating the efficiency of our proposed scheme.
{The time required to complete the training is 87 minutes.}
\begin{figure}[t]
	\centerline{\includegraphics[width=3in]{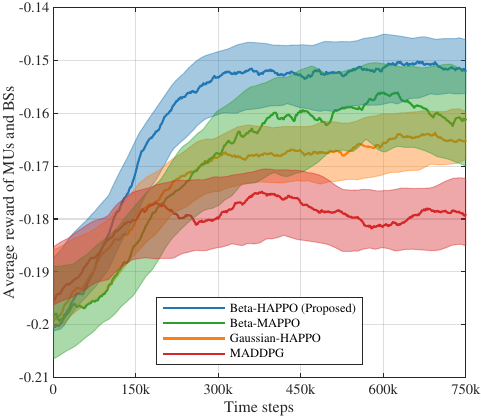}}
 \vspace{-1em}
	\caption{Average reward of MUs and BSs versus time steps.}
	\label{fig:con}
 \vspace{-1em}
\end{figure}

\begin{figure}[t]
	\centerline{\includegraphics[width=3in]{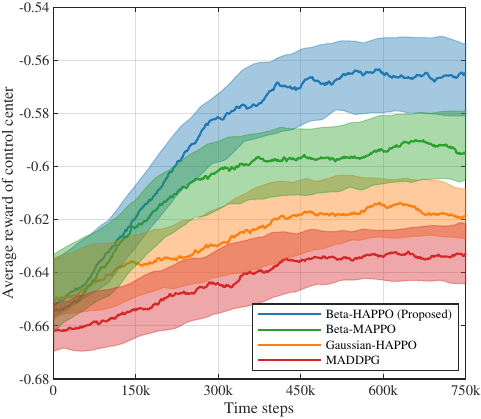}}
 \vspace{-1em}
	\caption{Average reward of control center versus time steps.}
	\label{fig:cenc}
 \vspace{-1em}
\end{figure}

\begin{figure}[t]
	\centerline{\includegraphics[width=3in]{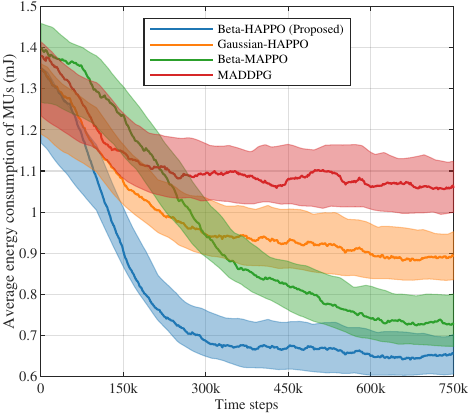}}
\vspace{-1em}
	\caption{Average energy consumption versus time steps.}
 \vspace{-1em}
	\label{fig:obj}
\end{figure}

\begin{figure}[t]
	\centerline{\includegraphics[width=3in]{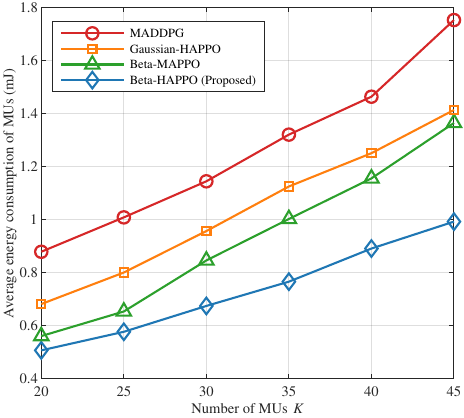}}
 \vspace{-1em}
	\caption{Average energy consumption versus number of MUs.}
	\label{fig:K-alg}
 \vspace{-1em}
\end{figure}

\begin{figure}[t]
	\centerline{\includegraphics[width=3in]{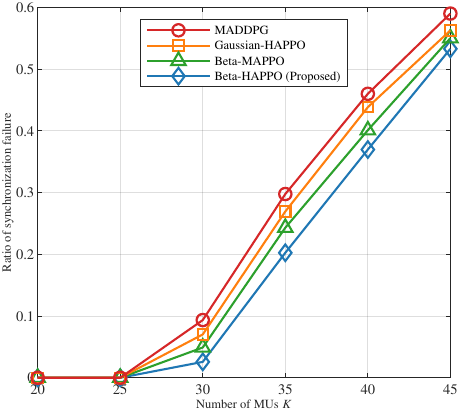}}
 \vspace{-1em}
	{\caption{Synchronization failure ratio versus number of MUs.} \label{fig:K-err}}
        
 \vspace{-1em}
\end{figure}

\begin{figure}[t]
	\centerline{\includegraphics[width=3in]{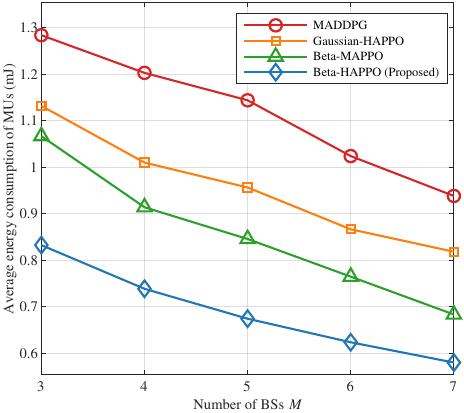}}
 \vspace{-1em}
	\caption{Average energy consumption versus number of BSs.}
	\label{fig:BS-alg}
 \vspace{-1em}
\end{figure}

\begin{figure}[t]
	\centerline{\includegraphics[width=3in]{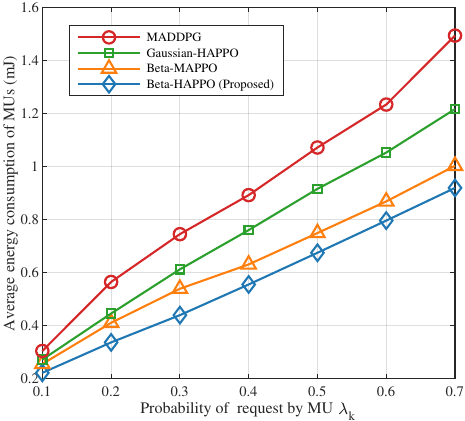}}
 \vspace{-1em}
	\caption{Average energy consumption versus probability of request.}
	\label{fig:prob-alg}
 \vspace{-1em}
\end{figure}

Fig. \ref{fig:K-alg} shows the average energy consumption versus the number of MUs.
It can be observed that the average energy consumption of MUs gradually increases as the number of MUs increases.
This is because the signal interference among MUs increases as their number increases. Consequently, the need for higher transmission power rises, resulting in an overall increase in average energy consumption.
The energy consumption of the proposed Beta-HAPPO scheme is the lowest, which verifies its effectiveness in minimizing the average energy compared to other schemes. 
{Fig. \ref{fig:K-err} shows the synchronization failure ratio versus number of MUs.
It can be observed that when the number of MUs is within 25, the synchronization failure ratio is 0; when the number of MUs exceeds 30, the synchronization failure ratio increases as the number of MUs increases. 
This occurs because when the number of MUs is small, the network resources are sufficient to meet the DT synchronization needs of the MUs. 
However, as the number of MUs increases, the failure ratio of DT synchronization gradually rises. This is because as more MUs join without an increase in the current network resources, resources become exhausted, leading to insufficient allocation for DT synchronization and inability to effectively complete DT synchronization.
The proposed Beta-HAPPO scheme has the lowest synchronization failure ratio, verifying its effectiveness in ensuring synchronization reliability compared to other schemes.}
Fig. \ref{fig:BS-alg} compares the average energy consumption of MUs under different numbers of BSs with $K=25$ MUs. We can see that the average energy consumption of MUs decreases as the number of BSs grows, which indicates that a larger number of BSs can offer a greater diversity gain. 
The proposed Beta-HAPPO retains its superiority on the benchmarks. 
Fig. \ref{fig:prob-alg} shows the average energy consumption versus the synchronization request probability $\lambda_k$. It is apparent that the average transmit energy increases with the request probability.
This is because with a higher synchronization request probability, MUs transmit synchronization data more frequently, resulting in an overall rise in average transmission energy.
Fig. \ref{fig:D-B} evaluates the impact of computational data size and bandwidth on the average energy consumption of MUs. As the maximum data size of synchronization data $D_{\max}$ increases from 0.2 Mb to 0.3 Mb, the average energy consumption of MUs constantly increases. In the meantime, the increase of bandwidth reduces the energy consumption.

Furthermore, we evaluate the impact of control factor $\eta$ on the energy consumption and synchronization failure ratio of MUs in Fig. \ref{fig:weight}. As $\eta$ increases from 0.7 to 1.3, it can be seen that the energy consumption of MUs increases, and the synchronization failure ratio of MUs decreases. 
This demonstrates the feasibility of incorporating the control factor to weigh the energy consumption of the MUs and the synchronization failure ratio.

\begin{figure}[t]
	\centerline{\includegraphics[width=3in]{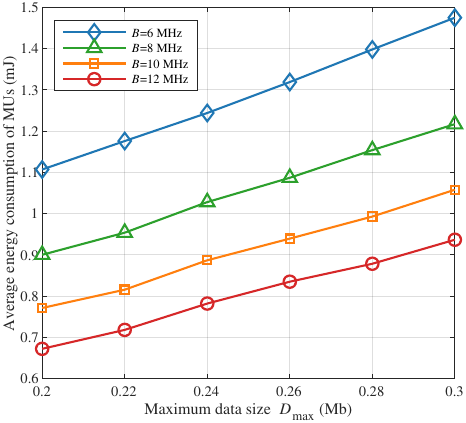}}
 \vspace{-1em}
	\caption{Average energy consumption under different bandwidth and data size.}
	\label{fig:D-B}
 \vspace{-1em}
\end{figure}

\begin{figure}[t]
	\centerline{\includegraphics[width=3in]{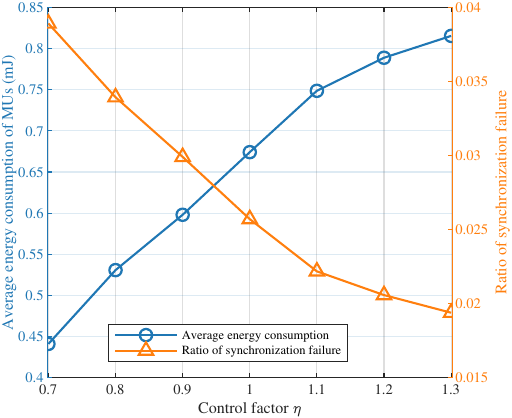}}
 \vspace{-1em}
	\caption{Average energy consumption and synchronization failure ratio of MUs versus control factor.}
	\label{fig:weight}
 \vspace{-1em}
\end{figure}

\begin{figure}[t]
	\centerline{\includegraphics[width=3in]{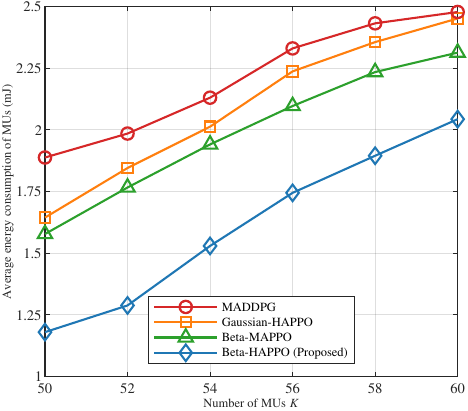}}
	\vspace{-1em}
	{\caption{Average energy consumption of MUs versus a large number of MUs.}
	\label{fig:large-user}}
	\vspace{-1em}
\end{figure}

\begin{figure}[t]
	\centerline{\includegraphics[width=3in]{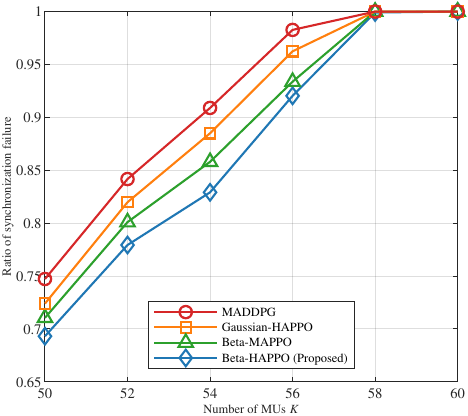}}
	\vspace{-1em}
	{\caption{Synchronization failure ratio versus a large number of MUs.}
	\label{fig:syn}}
	\vspace{-1em}
\end{figure}

{Finally, we evaluate the impact of a large number of MUs $K$ on the energy consumption and synchronization failure ratio of MUs in \ref{fig:large-user} and Fig. \ref{fig:syn}, respectively. As the number of users increases from 50 to 60, it can be seen from Fig. \ref{fig:large-user} that the energy consumption of MU increases.
Based on Fig. \ref{fig:syn},  when the number of MUs reaches 58, the DT synchronization failure ratio reaches  100\%, leading to a network outage. 
		This occurs because as more MUs join without any changes to the current network resources, the resources become exhausted, resulting in insufficient allocation for DT synchronization and inability to effectively complete DT synchronization.}

\section{Conclusion}\label{s:conclusion}
In this paper, we proposed a two-timescale framework for digital twin synchronization and migration. We minimized the long-term average energy consumption by jointly optimizing the communication and computation resource allocation under the reliability constraints of DT synchronization. In order to overcome the computational complexity brought by the non-convexity of the optimization model, we modeled the problem as a POMDP and proposed a two-timescale method based on HAPPO to solve the problem. 
Numerical results verified the convergence of the proposed Beta-HAPPO scheme and its capability in energy saving. 
In future work, we plan to integrate mixed energy sources into DT networks, aiming to achieve further reduction in system's energy consumption. 
{Furthermore, we also plan to conduct in-depth research on the impact of user scheduling and resource allocation on DT synchronization in scenarios with a large number of users and limited network resources.}

\section*{Appendices}
\subsection{Proof of Lemma 1}
\begin{proof}
	According to \eqref{eq:que} and $\left([x]^{+}\right)^2\leq x^2 $, we have
	\begin{equation} \label{eq:diffq}
		\begin{split}
			(Y_k[n+1])&^2-(Y_k[n])^2 \leq \left[Y_k[n]+X_k[n]- \varepsilon \right]^2-(Y_k[n])^2\\
			&\leq (X_k[n])^2 + \varepsilon^2+2Y_k[n]\left[X_k[n]-\varepsilon \right].
		\end{split}
	\end{equation}
	By summing up the results from \eqref{eq:diffq} over $n \in \{qT, qT+1,\cdots, (q+1)T-1\}$ and incorporating the conditional expectation with respect to $Y_k[qT]$, we can derive an upper bound for $\Delta \mathcal{L}_k[qT]$, which is given as follows:
	\begin{equation} \label{eq:diup}
		\begin{split}
			&\Delta \mathcal{L}_k[qT]\\
			&= \mathbb{E}\left\{\frac{1}{2}(Y_k[(q+1)T])^2-\frac{1}{2}(Y_k[qT])^2\bigg\vert Y_k[qT] \right\}\\
			&\leq \mathbb{E}\left\{\frac{1}{2}\sum \limits_{n\in\mathcal{T}_q}(X_k[n])^2+ \sum \limits_{n\in\mathcal{T}_q} (Y_k[n])[X_k[n]-\varepsilon]\bigg\vert Y_k[qT] \right\}\\
			&\quad +\frac{1}{2}\varepsilon^2T\\
			&\overset{(a)}{\le} \frac{1}{2}\mathbb{E}\left\{\sum \limits_{n\in\mathcal{T}_q}a_k[n]\right\}+\mathbb{E}\left\{\sum \limits_{n\in\mathcal{T}_q} (Y_k[n])[X_k[n]-\varepsilon]\bigg\vert Y_k[qT]\right\}\\
			&\quad +\frac{1}{2}\varepsilon^2T\\
			&\overset{(b)}{=}\frac{1}{2}(\lambda_k+\varepsilon^2)T+\mathbb{E}\left\{\sum \limits_{n\in\mathcal{T}_q} (Y_k[n])[X_k[n]-\varepsilon]\bigg\vert Y_k[qT]\right\},
		\end{split}
	\end{equation}
where $(a)$ holds because $(X_k[n])^2=X_k[n]\leq a_k[n]$ and $a_k[n]$ is independent of $Y_k[qT]$. $(b)$ holds because $a_k[n]$ is i.i.d over slot with $\mathbb{E}\left\{a_k[n]\right\}=\lambda_k$. Finally, let $B_{k,1}=(\lambda_k+\varepsilon^2)T/2$, we can yield the result \eqref{eq:diup}. 
\vspace{-1em}
\end{proof}

\subsection{Proof of Lemma 2}
\begin{proof}
	Since $X_k[n]\in \{0,1\}$, the queue length $Y_k[n]$ for each time slot $n$ is bounded by
	\begin{equation} \label{eq:bound}
		Y_k[qT]-(n-qT)\varepsilon \leq Y_k[n] \leq Y_k[qT]+(n-qT)(1-\varepsilon).
	\end{equation}
Using \eqref{eq:bound}, it can be observed that the term $\mathbb{E}\{\sum \limits_{n\in\mathcal{T}_q}Y_k[n][X_k[n]-\varepsilon]\vert Y_k[qT] \}$ in \eqref{eq:up1} can be bounded as
	\begin{equation} \label{eq:lee2}
		\begin{split}
			&\sum \limits_{n\in\mathcal{T}_q}Y_k[n]\left[X_k[n]-\varepsilon \right]\\
			&\leq \sum \limits_{n\in\mathcal{T}_q}\left[Y_k[qT]+(n-qT)(1-\varepsilon)\right]Y_k[n]\\
			&\quad -\sum \limits_{n\in\mathcal{T}_q}[Y_k[qT]-(n-qT)\varepsilon]\varepsilon\\
			&\leq \sum \limits_{n\in\mathcal{T}_q}Y_k[kT][X_k[n]-\varepsilon] -\sum \limits_{n\in\mathcal{T}_q}(n-qT)[(1-\varepsilon)X_k[n]+\varepsilon^2].
		\end{split}
	\end{equation}
Given $Y_k[qT]$, applying the conditional expectation to \eqref{eq:lee2} yields the following expression:
	\begin{equation} \label{eq:lee3}
		\begin{split}
			&\mathbb{E}\left\{\sum \limits_{n\in\mathcal{T}_q} (Y_k[n])[X_k[n]-\varepsilon]\bigg\vert Y_k[qT]\right\}\\
			&\leq \mathbb{E}\left\{\sum \limits_{n\in\mathcal{T}_q}Y_k[kT][X_k[n]-\varepsilon]\bigg\vert Y_k[qT]\right\}\\
			&\quad -\mathbb{E}\left\{\sum \limits_{n\in\mathcal{T}_q}(n-qT)[(1-\varepsilon)X_k[n]+\varepsilon^2]\bigg\vert Y_k[qT]\right\}\\
			&\leq\mathbb{E}\left\{\sum \limits_{n\in\mathcal{T}_q}Y_k[kT][X_k[n]-\varepsilon]\bigg\vert Y_k[qT]\right\}\\
			&\quad +\frac{1}{2}T(T-1)[(1-\varepsilon)\lambda_k+\varepsilon^2].
		\end{split}
	\end{equation}
Based on  \eqref{eq:lee3} and setting $B_{k,2}\triangleq B_{k,1}+(T-1)[(1-\varepsilon)\lambda_k+\varepsilon^2]/2$, we can additionally relax the inequality \eqref{eq:up1} into \eqref{eq:up2}, which completes the proof. 
\end{proof}

\bibliographystyle{IEEEtran}
\bibliography{IEEEabrv,refs}

\begin{thebibliography}{10}
\providecommand{\url}[1]{#1}
\csname url@samestyle\endcsname
\providecommand{\newblock}{\relax}
\providecommand{\bibinfo}[2]{#2}
\providecommand{\BIBentrySTDinterwordspacing}{\spaceskip=0pt\relax}
\providecommand{\BIBentryALTinterwordstretchfactor}{4}
\providecommand{\BIBentryALTinterwordspacing}{\spaceskip=\fontdimen2\font plus
\BIBentryALTinterwordstretchfactor\fontdimen3\font minus
  \fontdimen4\font\relax}
\providecommand{\BIBforeignlanguage}[2]{{%
\expandafter\ifx\csname l@#1\endcsname\relax
\typeout{** WARNING: IEEEtran.bst: No hyphenation pattern has been}%
\typeout{** loaded for the language `#1'. Using the pattern for}%
\typeout{** the default language instead.}%
\else
\language=\csname l@#1\endcsname
\fi
#2}}
\providecommand{\BIBdecl}{\relax}
\BIBdecl

\bibitem{Yang2019IN_6G}
P.~Yang, Y.~Xiao, M.~Xiao, and S.~Li, ``6{G} wireless communications: {Vision}
  and potential techniques,'' \emph{{IEEE} Netw.}, vol.~33, no.~4, pp. 70--75,
  Jul. 2019.

\bibitem{Letaief2022JSAC}
K.~B. Letaief, Y.~Shi, J.~Lu, and J.~Lu, ``Edge artificial intelligence for
  6{G}: {V}ision, enabling technologies, and applications,'' \emph{{IEEE} J.
  Sel. Areas Commun.}, vol.~40, no.~1, pp. 5--36, Jan. 2022.

\bibitem{yfu2023}
Y.~Fu, Y.~Shan, Q.~Zhu, K.~Hung, Y.~Wu, and T.~Q.~S. Quek, ``A distributed
  microservice-aware paradigm for {6G}: {C}hallenges, principles, and research
  opportunities,'' \emph{{IEEE} Netw.}, vol.~3, no.~38, pp. 163--170, May 2024.

\bibitem{Wu2021IOTJ_Digital}
Y.~Wu, K.~Zhang, and Y.~Zhang, ``Digital twin networks: {A} survey,''
  \emph{{IEEE} Internet Things J.}, vol.~8, no.~18, pp. 13\,789--13\,804, Sep.
  2021.

\bibitem{Wang2022TCOMR}
H.~Wang, C.~Liu, Z.~Shi, Y.~Fu, and R.~Song, ``Power minimization for uplink
  {RIS}-assisted {CoMP-NOMA} networks with {GSIC},'' \emph{{IEEE} Trans.
  Commun.}, vol.~70, no.~7, pp. 4559--4573, Jul. 2022.

\bibitem{Tang2022OJCOMS}
F.~Tang, X.~Chen, T.~K. Rodrigues, M.~Zhao, and N.~Kato, ``Survey on digital
  twin edge networks ({DITEN}) toward 6{G},'' \emph{IEEE Open J. Commun. Soc.},
  vol.~3, pp. 1360--1381, 2022.

\bibitem{Wang2023IOTJ1}
Y.~Wang, Z.~Su, S.~Guo, M.~Dai, T.~H. Luan, and Y.~Liu, ``A survey on digital
  twins: {A}rchitecture, enabling technologies, security and privacy, and
  future prospects,'' \emph{{IEEE} Internet Things J.}, vol.~10, no.~17, pp.
  14\,965--14\,987, Sep. 2023.

\bibitem{Gu2024TVT}
J.~Gu, Y.~Fu, and K.~Hung, ``On intelligent placement decision-making
  algorithms for wireless digital twin networks via bandit learning,''
  \emph{{IEEE} Trans. Veh. Technol.}, vol.~73, no.~6, pp. 8889--8902, Jun.
  2024.

\bibitem{Lu2021IOTJ}
Y.~Lu, S.~Maharjan, and Y.~Zhang, ``Adaptive edge association for wireless
  digital twin networks in 6{G},'' \emph{{IEEE} Internet Things J.}, vol.~8,
  no.~22, pp. 16\,219--16\,230, Nov. 2021.

\bibitem{Alcaraz2022COMST_Digital}
C.~Alcaraz and J.~Lopez, ``Digital twin: {A} comprehensive survey of security
  threats,'' \emph{IEEE Commun. Surv. Tutor.}, vol.~24, no.~3, pp. 1475--1503,
  3rd Quart. 2022.

\bibitem{Wang2022IoT_Mobility}
Z.~Wang, R.~Gupta, K.~Han, H.~Wang, A.~Ganlath, N.~Ammar, and P.~Tiwari,
  ``Mobility digital twin: {Concept}, architecture, case study, and future
  challenges,'' \emph{{IEEE} Internet Things J.}, vol.~9, no.~18, pp.
  17\,452--17\,467, Sep. 2022.

\bibitem{Li2024TMC}
J.~Li, S.~Guo, W.~Liang, J.~Wang, Q.~Chen, Y.~Zeng, B.~Ye, and X.~Jia,
  ``Digital twin-enabled service provisioning in edge computing via continual
  learning,'' \emph{{IEEE} Trans. Mobile Comput.}, vol.~23, no.~6, pp.
  7335--7350, Jun. 2024.

\bibitem{Duong2023MWC}
T.~Q. Duong, D.~Van~Huynh, S.~R. Khosravirad, V.~Sharma, O.~A. Dobre, and
  H.~Shin, ``From digital twin to metaverse: {T}he role of 6{G} ultra-reliable
  and low-latency communications with multi-tier computing,'' \emph{{IEEE}
  Wireless Commun.}, vol.~30, no.~3, pp. 140--146, Jun. 2023.

\bibitem{Zhang2023JSAC}
Y.~Zhang, J.~Hu, and G.~Min, ``Digital twin-driven intelligent task offloading
  for collaborative mobile edge computing,'' \emph{{IEEE} J. Sel. Areas
  Commun.}, vol.~41, no.~10, pp. 3034--3045, Oct. 2023.

\bibitem{Khan2022MCOM}
L.~U. Khan, W.~Saad, D.~Niyato, Z.~Han, and C.~S. Hong, ``Digital-twin-enabled
  6{G}: {V}ision, architectural trends, and future directions,'' \emph{{IEEE}
  Commun. Mag.}, vol.~60, no.~1, pp. 74--80, Jan. 2022.

\bibitem{Kurma2023TCOM}
S.~Kurma, M.~Katwe, K.~Singh, C.~Pan, S.~Mumtaz, and C.-P. Li,
  ``{RIS}-empowered {MEC} for {URLLC} systems with digital-twin-driven
  architecture,'' \emph{{IEEE} Trans. Commun.}, vol.~72, no.~4, pp. 1983--1997,
  Apr. 2024.

\bibitem{Dai2021TII}
Y.~Dai, K.~Zhang, S.~Maharjan, and Y.~Zhang, ``Deep reinforcement learning for
  stochastic computation offloading in digital twin networks,'' \emph{{IEEE}
  Trans. Ind. Informat.}, vol.~17, no.~7, pp. 4968--4977, Jul. 2021.

\bibitem{Li2022TVT}
B.~Li, Y.~Liu, L.~Tan, H.~Pan, and Y.~Zhang, ``Digital twin assisted task
  offloading for aerial edge computing and networks,'' \emph{{IEEE} Trans. Veh.
  Technol.}, vol.~71, no.~10, pp. 10\,863--10\,877, Oct. 2022.

\bibitem{Van2022TCOM}
D.~Van~Huynh, V.-D. Nguyen, S.~R. Khosravirad, V.~Sharma, O.~A. Dobre, H.~Shin,
  and T.~Q. Duong, ``{URLLC} edge networks with joint optimal user association,
  task offloading and resource allocation: {A} digital twin approach,''
  \emph{{IEEE} Trans. Commun.}, vol.~70, no.~11, pp. 7669--7682, Nov. 2022.

\bibitem{Van2023JSAC}
D.~Van~Huynh, V.-D. Nguyen, S.~R. Khosravirad, G.~K. Karagiannidis, and T.~Q.
  Duong, ``Distributed communication and computation resource management for
  digital twin-aided edge computing with short-packet communications,''
  \emph{{IEEE} J. Sel. Areas Commun.}, vol.~41, no.~10, pp. 3008--3021, Oct.
  2023.

\bibitem{Xu2023JSAC}
C.~Xu, Z.~Tang, H.~Yu, P.~Zeng, and L.~Kong, ``Digital twin-driven
  collaborative scheduling for heterogeneous task and edge-end resource via
  multi-agent deep reinforcement learning,'' \emph{{IEEE} J. Sel. Areas
  Commun.}, vol.~41, no.~10, pp. 3056--3069, Oct. 2023.

\bibitem{Liu2022IOTJ}
T.~Liu, L.~Tang, W.~Wang, Q.~Chen, and X.~Zeng, ``Digital-twin-assisted task
  offloading based on edge collaboration in the digital twin edge network,''
  \emph{{IEEE} Internet Things J.}, vol.~9, no.~2, pp. 1427--1444, Jan. 2022.

\bibitem{Guo2024MWC}
Q.~Guo, F.~Tang, T.~K. Rodrigues, and N.~Kato, ``Five disruptive technologies
  in 6{G} to support digital twin networks,'' \emph{{IEEE} Wireless Commun.},
  vol.~31, no.~1, pp. 149--155, Jan. 2024.

\bibitem{Yu2024IOTM}
J.~Yu, A.~Alhilal, P.~Hui, and D.~H.~K. Tsang, ``Bi-directional digital twin
  and edge computing in the metaverse,'' \emph{IEEE Internet Things Mag.},
  vol.~7, no.~3, pp. 106--112, May 2024.

\bibitem{Han2023JIOT}
Y.~Han, D.~Niyato, C.~Leung, D.~I. Kim, K.~Zhu, S.~Feng, X.~Shen, and C.~Miao,
  ``A dynamic hierarchical framework for {IoT}-assisted digital twin
  synchronization in the metaverse,'' \emph{{IEEE} Internet Things J.},
  vol.~10, no.~1, pp. 268--284, Jan. 2023.

\bibitem{Vaezi2023IOTJ}
M.~Vaezi, K.~Noroozi, T.~D. Todd, D.~Zhao, and G.~Karakostas, ``Digital twin
  placement for minimum application request delay with data age targets,''
  \emph{{IEEE} Internet Things J.}, vol.~10, no.~13, pp. 11\,547--11\,557, Jul.
  2023.

\bibitem{Zheng2023TWC}
J.~Zheng, T.~H. Luan, Y.~Zhang, R.~Li, Y.~Hui, L.~Gao, and M.~Dong, ``Data
  synchronization in vehicular digital twin network: {A} game theoretic
  approach,'' \emph{{IEEE} Trans. Wireless Commun.}, vol.~22, no.~11, pp.
  7635--7647, Nov. 2023.

\bibitem{Chukhno2022IOTJ}
O.~Chukhno, N.~Chukhno, G.~Araniti, C.~Campolo, A.~Iera, and A.~Molinaro,
  ``Placement of social digital twins at the edge for beyond {5G IoT}
  networks,'' \emph{{IEEE} Internet Things J.}, vol.~9, no.~23, pp.
  23\,927--23\,940, Dec. 2022.

\bibitem{Liang2003TNET}
B.~Liang and Z.~Haas, ``Predictive distance-based mobility management for
  multidimensional {PCS} networks,'' \emph{{IEEE/ACM} Trans. Netw.}, vol.~11,
  no.~5, pp. 718--732, Oct. 2003.

\bibitem{Batabyal2015COMST}
S.~Batabyal and P.~Bhaumik, ``Mobility models, traces and impact of mobility on
  opportunistic routing algorithms: {A} survey,'' \emph{IEEE Commun. Surv.
  Tutor.}, vol.~17, no.~3, pp. 1679--1707, 3rd Quart. 2015.

\bibitem{Liu2020TVT}
Q.~Liu, L.~Shi, L.~Sun, J.~Li, M.~Ding, and F.~Shu, ``Path planning for
  {UAV}-mounted mobile edge computing with deep reinforcement learning,''
  \emph{{IEEE} Trans. Veh. Technol.}, vol.~69, no.~5, pp. 5723--5728, May 2020.

\bibitem{Hua2021TWC}
M.~Hua, L.~Yang, Q.~Wu, C.~Pan, C.~Li, and A.~L. Swindlehurst, ``{UAV}-assisted
  intelligent reflecting surface symbiotic radio system,'' \emph{{IEEE} Trans.
  Wireless Commun.}, vol.~20, no.~9, pp. 5769--5785, Sep. 2021.

\bibitem{Wu2021TII}
W.~Wu, P.~Yang, W.~Zhang, C.~Zhou, and X.~Shen, ``Accuracy-guaranteed
  collaborative {DNN} inference in industrial {IoT} via deep reinforcement
  learning,'' \emph{{IEEE} Trans. Ind. Informat.}, vol.~17, no.~7, pp.
  4988--4998, Jul. 2021.

\bibitem{Xu2022JSAC}
X.~Xu, Q.~Chen, X.~Mu, Y.~Liu, and H.~Jiang, ``Graph-embedded multi-agent
  learning for smart reconfigurable {THz MIMO-NOMA} networks,'' \emph{{IEEE} J.
  Sel. Areas Commun.}, vol.~40, no.~1, pp. 259--275, Jan. 2022.

\bibitem{neely2010stochastic}
M.~J. Neely, ``Stochastic network optimization with application to
  communication and queueing systems,'' \emph{Synthesis Lectures Commun.
  Netw.}, vol.~3, no.~1, pp. 1--211, 2010.

\bibitem{Yao2014TPDS}
Y.~Yao, L.~Huang, A.~B. Sharma, L.~Golubchik, and M.~J. Neely, ``Power cost
  reduction in distributed data centers: {A} two-time-scale approach for delay
  tolerant workloads,'' \emph{{IEEE} Trans. Parallel Distrib. Syst.}, vol.~25,
  no.~1, pp. 200--211, Jan. 2014.

\bibitem{Liang2022TWC}
Z.~Liang, Y.~Liu, T.-M. Lok, and K.~Huang, ``A two-timescale approach to
  mobility management for multicell mobile edge computing,'' \emph{{IEEE}
  Trans. Wireless Commun.}, vol.~21, no.~12, pp. 10\,981--10\,995, Dec. 2022.

\bibitem{lowe2017multi}
R.~Lowe, Y.~I. Wu, A.~Tamar, J.~Harb, O.~Pieter~Abbeel, and I.~Mordatch,
  ``Multi-agent actor-critic for mixed cooperative-competitive environments,''
  in \emph{Proc. Adv. Neural Inf. Process. Syst.}, 2017.

\bibitem{Jakub2022ICLR_Trust}
J.~G. Kuba, R.~Chen, M.~Wen, Y.~Wen, F.~Sun, J.~Wang, and Y.~Yang, ``Trust
  region policy optimisation in multi-agent reinforcement learning,'' in
  \emph{Int. Conf. Learn. Represent.}, 2022.

\bibitem{Cai2022TNSE}
T.~Cai, Z.~Yang, Y.~Chen, W.~Chen, Z.~Zheng, Y.~Yu, and H.-N. Dai,
  ``Cooperative data sensing and computation offloading in {UAV}-assisted
  crowdsensing with multi-agent deep reinforcement learning,'' \emph{{IEEE}
  Trans. Netw. Sci. Eng.}, vol.~9, no.~5, pp. 3197--3211, Sep.- Oct. 2022.

\bibitem{Schulman2015arxiV}
\BIBentryALTinterwordspacing
J.~Schulman, P.~Moritz, S.~Levine, M.~Jordan, and P.~Abbeel, ``High-dimensional
  continuous control using generalized advantage estimation,'' \emph{arXiv
  preprint arXiv:1506.02438}, 2015. [Online]. Available:
  \url{https://arxiv.org/abs/1506.02438}
\BIBentrySTDinterwordspacing

\bibitem{Grondman2012TSMCC}
I.~Grondman, L.~Busoniu, G.~A.~D. Lopes, and R.~Babuska, ``A survey of
  actor-critic reinforcement learning: {S}tandard and natural policy
  gradients,'' \emph{{IEEE} Trans. Syst., Man, Cybern. {C}}, vol.~42, no.~6,
  pp. 1291--1307, Nov. 2012.

\bibitem{Yu2021ar_Surprising}
\BIBentryALTinterwordspacing
C.~Yu, A.~Velu, E.~Vinitsky, Y.~Wang, A.~Bayen, and Y.~Wu, ``The surprising
  effectiveness of {PPO} in cooperative, multi-agent games,'' \emph{arXiv
  preprint arXiv:2103.01955}, 2021. [Online]. Available:
  \url{http://arxiv.org/abs/2103.01955}
\BIBentrySTDinterwordspacing

\bibitem{Waqar2022TITS}
N.~Waqar, S.~A. Hassan, A.~Mahmood, K.~Dev, D.-T. Do, and M.~Gidlund,
  ``Computation offloading and resource allocation in {MEC}-enabled integrated
  aerial-terrestrial vehicular networks: {A} reinforcement learning approach,''
  \emph{{IEEE} Trans. Intell. Transp. Syst.}, vol.~23, no.~11, pp.
  21\,478--21\,491, Jun. 2022.

\bibitem{Liu2023TWCU}
W.~Liu, B.~Li, W.~Xie, Y.~Dai, and Z.~Fei, ``Energy efficient computation
  offloading in aerial edge networks with multi-agent cooperation,''
  \emph{{IEEE} Trans. Wireless Commun.}, vol.~22, no.~9, pp. 5725--5739, Sep.
  2023.

\end{thebibliography}

\end{document}